\documentclass[12pt]{extarticle}
\usepackage[utf8]{inputenc}

\usepackage{amsmath, amsthm, amssymb, color, url}
\usepackage{graphicx}
\usepackage{natbib}

\usepackage{mathtools}
\usepackage{multirow}
\usepackage{booktabs}

\usepackage[shortlabels]{enumitem}
\usepackage{makecell}
\usepackage[version=4]{mhchem}

\usepackage{caption}
\usepackage{subcaption}

\usepackage{tikz}
\usepackage{pgfplots} 
\usepgfplotslibrary{fillbetween}
\pgfplotsset{compat=1.14}

\usepackage[colorlinks=true, allcolors=blue]{hyperref}

\oddsidemargin \evensidemargin
\newcommand{\h}{\mathfrak{h}}

\newcommand{\koff}{k_{\mathrm{off}}}
\newcommand{\kon}{k_{\mathrm{on}}}
\newcommand{\kcat}{k_{\mathrm{cat}}}
\newcommand{\lcat}{{\ell}_{\mathrm{cat}}}
\newcommand{\loff}{{\ell}_{\mathrm{off}}}
\newcommand{\lon}{{\ell}_{\mathrm{on}}}

\newcommand{\pk}{p_k}
\newcommand{\pl}{p_{\ell}}

\newcommand{\R}{\mathbb{R}}

\theoremstyle{plain}
\newtheorem{theorem}{Theorem}[section]
\newtheorem{corollary}[theorem]{Corollary}
\newtheorem{proposition}[theorem]{Proposition}

\newtheorem{conjecture}[theorem]{Conjecture}
\newtheorem{question}[theorem]{Question}

\theoremstyle{definition}
\newtheorem{defn}[theorem]{Definition}

\newtheorem{remark}[theorem]{Remark}

\newcommand{\pd}[2]{\ensuremath{ \frac{ \partial #1 }{\partial #2} }}

\newcommand{\scc}{\mathcal{S}}
\newcommand{\F}{f_{c, \kappa}}

\newcommand{\St}{{S}}

\newcommand{\defword}[1]{\textcolor{purple}{\underline{\textbf{#1}}}}

\usepackage[affil-it]{authblk}

\definecolor{dblue}{rgb}{0.0,0.0,0.68}

\title{Dynamics of ERK regulation in the processive limit}
\author[1]{Carsten Conradi}
\author[2]{Nida Obatake}
\author[2]{Anne Shiu}
\author[2,3]{Xiaoxian Tang}
\affil[1]{HTW Berlin, Germany}
\affil[2]{Department of Mathematics, Texas A\&M University, USA}
\affil[3]{School of Mathematical Sciences, Beihang University, China}

\date{\today}

\begin{document}

\maketitle

\begin{abstract}
We consider a model of extracellular signal-regulated kinase (ERK) regulation by dual-site phosphorylation and dephosphorylation, which exhibits bistability and oscillations, but loses these properties in the limit in which
the mechanisms underlying
phosphorylation and dephosphorylation become processive.
Our results suggest that anywhere along the way to becoming processive, the model remains bistable and oscillatory.
More precisely, in simplified versions of the model,
precursors to bistability and oscillations (specifically,
multistationarity and Hopf bifurcations, respectively) exist at all ``processivity levels''.  Finally, we investigate whether bistability and oscillations can exist together.
\end{abstract}

\section{Introduction} \label{sec:intro}
We focus on the following question, posed by \cite{long-term}, pertaining to a model of extracellular signal-regulated kinase (ERK) regulation (Figure~\ref{fig:erk_network}):

\begin{question} \label{q:main}
For all \defword{processivity levels}\footnote{This level is the probability that the enzyme acts processively, that is, adds a second phosphate group after adding the first~\citep{salazar}.  
A somewhat similar idea, from~\citep{sun2014enhancement}, 
is the ``degree of processivity''. }
$p_k:={\kcat}/{(\kcat+\koff)}$ and $p_{\ell}:={\lcat}/{(\lcat+\loff)}$ close
to 1, is the ERK network in Figure~\ref{fig:erk_network},
bistable and oscillatory?
\end{question}

\begin{figure}[ht]
    \centering
   \includegraphics[trim=170 270 170 327,clip,scale=1.2]{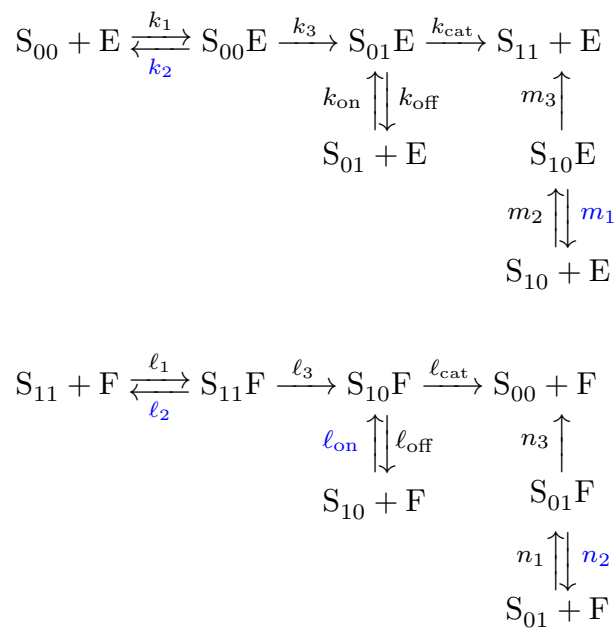}
   \caption{The \defword{ERK network} consists of ERK regulation through dual-site phosphorylation by the kinase MEK
(denoted by $E$) and dephosphorylation by the phosphatase MKP3 ($F$).  Each $S_{ij}$ denotes an ERK phosphoform, with subscripts indicating at which of two sites phosphate groups are attached.  Deleting from this network the reactions labeled $k_2, m_1, l_2, \lon, n_2$ (in blue) yields the \defword{minimally bistable ERK subnetwork} (the explanation for this name is given before Question~\ref{q:subnet-bistable}).
   }
   \label{fig:erk_network}
\end{figure}

The motivation behind this question 
was given earlier~\citep{  erk-model,OSTT, long-term}, which we summarize here.
Briefly, as both $p_k$ and $p_{\ell}$ approach~1, the ERK network
``limits'' to a (fully processive) network that is globally convergent
to a unique steady state, and thus lacks bistability and oscillations \citep{ConradiShiu}.  As bistability and oscillations may allow networks to act as a biological switch or clock~\citep{tyson-albert}, we want to know how far ``along the way'' to the limit, the network maintains the capacity for these important dynamical properties.

A partial result toward resolving Question~\ref{q:main} was given by
\cite{long-term},
who exhibited, in simulations, oscillations for $p_k,\pl \approx 0.97$. 
This left open the question of oscillations for $0.97<p_k,\pl<1$.
Our result in this direction is given in Theorem~\ref{thm:pre-hopf-all-process-levels-epsilon-close-to-1} (described below).

Additional prior results aimed at answering Question~\ref{q:main} appeared in work of three of the present authors with Torres~\citep{OSTT}.
We showed that bistability is preserved when reactions in the ERK network are made irreversible, as long at least one of the reactions labeled by $\kon$ and $\lon$ is preserved.  We therefore give the name ``minimally bistable ERK subnetwork'' to the network obtained by making all reaction irreversible except the reversible-reaction pair $\kon$ and $\koff$
(Figure~\ref{fig:erk_network}).  (By symmetry, the network preserving $\lon$ and $\loff$, rather than $\kon$ and $\koff$, is equivalent.)  We therefore state the following version of Question~\ref{q:main} for bistability:
\begin{question}    \label{q:subnet-bistable}
For $p_k$ and $p_{\ell}$ close
to 1, is the minimally bistable ERK subnetwork, bistable?
\end{question}
\noindent
If yes, then by results lifting bistability from subnetworks to larger networks~\citep{Joshi-Shiu-2013},
this also answers in the affirmative the part of
Question~\ref{q:main} pertaining to bistability.

\begin{figure}[ht]
    \centering
    \includegraphics[scale=1.2]{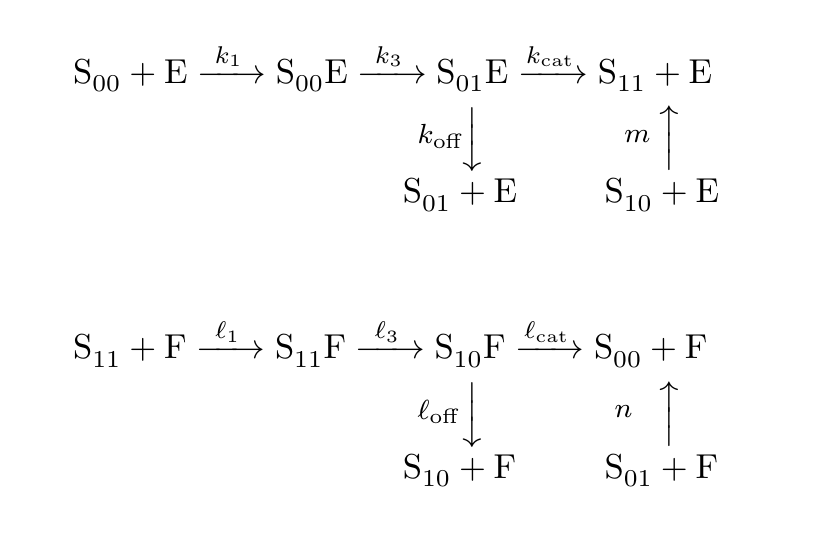}
    \caption{\defword{Reduced ERK network}~\citep{OSTT}}
    \label{fig:red_erk_net}
\end{figure}

Similarly,
for oscillations, we showed that when reactions are made irreversible
and also two ``intermediates'' (namely, $S_{10}E$ and $S_{01}F$)
are removed, oscillations are preserved~\citep{OSTT}.
For this network, called the ``reduced ERK network'' (Figure~\ref{fig:red_erk_net}), we now
ask a variant of Question~\ref{q:main} for oscillations
(an affirmative answer to Question~\ref{q:reduced-osc} likely ``lifts'' to an affirmative answer to Question~\ref{q:main};
see Remark~\ref{rmk:reln-to-orig-q}):
\begin{question} \label{q:reduced-osc}
For $p_k$ and $p_{\ell}$ close
to 1, is the reduced ERK network,
oscillatory?
\end{question}

Our answers to Questions~\ref{q:subnet-bistable} and~\ref{q:reduced-osc} are as follows.
For the first question, at {\em all} processivity levels -- not just near 1 -- the minimally bistable ERK subnetwork admits multiple steady states, a necessary condition for bistability (Theorem~\ref{thm:MSS}).  Furthermore, computational evidence suggests that indeed we have bistability.
{\color{black} We also investigate how varying processivity levels affects the range of parameter values that yield multistationarity and also \emph{how} multistationarity (and thus bistability) is lost as
the ERK network limits to a (fully processive) network without bistability.
Our numerical observations suggest that as the processivity levels approach 1,
the classical S-shaped curve often associated with multistationarity deforms to a steep Hill function (see 
Figure~\ref{fig:pk_pl_099} in Section~\ref{sec:numerics}). }

Similarly, for
Question~\ref{q:reduced-osc},
again at (nearly) all processivity levels, the reduced ERK network admits a Hopf bifurcation (Theorem~\ref{thm:pre-hopf-all-process-levels-epsilon-close-to-1}), a precursor to oscillations.
We also numerically investigate such oscillations (see Figure~\ref{fig:osci}).

Finally, we pursue several more questions pertaining to ERK networks.
We investigate in the ERK network whether -- for some choice of rate constants -- bistability and Hopf bifurcations can coexist (see Theorem~\ref{thm:no-coexistence}).
We also pursue a conjecture of~\cite{OSTT} on the maximum number of steady states in the minimally bistable ERK network.

Our results fit into related literature as follows.
First, as other authors have done for their models of interest~\citep{CFM,GRPD,SF}, we analyze simplified versions of the ERK network obtained by removing intermediate species and/or reactions (in some cases, bistability and oscillations can be ``lifted'' from smaller networks to larger ones~\citep{banaji-inheritance, BP-inher, add-flow, fw2013, Joshi-Shiu-2013}).
Also, our proofs harness two results from previous work:
a Hopf-bifurcation criterion for the reduced ERK network~\citep{OSTT}, and a criterion for multistationarity arising from degree theory~\citep{CFMW,DPST}.

This work proceeds as follows.
Section~\ref{sec:bkrd} provides background on chemical reaction systems and other topics.
In Section~\ref{sec:networks}, we give some details about the networks we study.
Next, we present our main results on multistationarity and bistability (Section~\ref{sec:bistable}),
Hopf bifurcations and oscillations (Section~\ref{sec:osc}), and coexistence of bistability and oscillations (Section~\ref{sec:coexist}).
In Section~\ref{sec:number-steady-states}, we prove results on the maximum number of steady states in the minimally bistable ERK network.
We conclude with a Discussion in Section~\ref{sec:disc}.

\section{Background} \label{sec:bkrd}
This section contains background on chemical reaction systems and their steady states.
We also recall how ``steady-state parametrizations'' can be used to assess whether a network is multistationary (Proposition~\ref{prop:hopfcriterion}).

\subsection{Chemical reaction systems} \label{sec:CRS}
As in~\citep{DPST}, our notation closely matches that of \cite{CFMW}.
A \defword{reaction network} $G$  (or, for brevity, {\em network}) consists of a set of $s$ species $\{X_1, X_2, \ldots, X_s\}$ and a set of $m$ reactions:
\[
\alpha_{1j}X_1 +
\alpha_{2j}X_2 +  \dots +
\alpha_{sj}X_s
~ \to ~
\beta_{1j}X_1 +
\beta_{2j}X_2 +  \dots +
\beta_{sj}X_s~,
 \quad \quad
    {\rm for}~
	j=1,2, \ldots, m~,
\]
where each $\alpha_{ij}$ and $\beta_{ij}$ is a non-negative integer. The \defword{stoichiometric matrix} of
$G$,
denoted by $N$, is the $s\times m$ matrix with
$N_{ij}= \beta_{ij}-\alpha_{ij}$.
Let $d=s-{\rm rank}(N)$. 
The image of $N$ is the
\defword{stoichiometric subspace}, denoted by $S$.
A \defword{conservation-law matrix} of $G$, denoted by $W$, is a row-reduced $d\times s$-matrix
such that the rows form a basis of
{\color{black} the orthogonal complement of $S$.}
If there exists a choice of $W$ such that each entry is nonnegative and each column
contains at least one nonzero entry (equivalently, each species occurs in at least one nonnegative conservation law), then $G$ is \defword{conservative}.

Denote the concentrations of the species $X_1,X_2, \ldots, X_s$ by $x_1, x_2, \ldots, x_s$, respectively.
These concentrations, under the assumption of {\em mass-action kinetics}, evolve according to the following system of ODEs:
\begin{equation}\label{sys}
\dot{x}~=~f(x)~:=~N\cdot \begin{pmatrix}
\kappa_1 \, x_1^{\alpha_{11}}
		x_2^{\alpha_{21}}
		\cdots x_s^{\alpha_{s1}} \\
\kappa_2 \, x_1^{\alpha_{12}}
		x_2^{\alpha_{22}}
		\cdots x_s^{\alpha_{s2}} \\
		\vdots \\
\kappa_m \, x_1^{\alpha_{1m}}
		x_2^{\alpha_{2m}}
		\cdots x_s^{\alpha_{sm}} \\
\end{pmatrix}~,
\end{equation}
where $x=(x_1, x_2, \ldots, x_s)$,
and each $\kappa_j \in \mathbb R_{>0}$ is a \defword{reaction rate constant}.
 By considering the rate constants as a vector of parameters $\kappa=(\kappa_1, \kappa_2, \dots, \kappa_m)$, we have polynomials $f_{\kappa,i} \in \mathbb Q[\kappa,x]$, for $i=1,2, \dots, s$.
For ease of notation, we often write $f_i$ rather than  $f_{\kappa,i}$.

A solution $x(t)$ with nonnegative initial values $x(0)=x^0 \in
\mathbb{R}^s_{\geq 0}$ remains, for all positive time,
 in the following \defword{stoichiometric compatibility class} with respect to the \defword{total-constant vector} $c\coloneqq W x^0 \in {\mathbb R}^d$~{\color{blue} \citep[Lemma II.3]{Sontag01}}: 
\begin{align} \label{eqn:invtPoly}
\scc_c~\coloneqq~ \{x\in {\mathbb R}_{\geq 0}^s \mid Wx=c\}~.
\end{align}
A \defword{steady state} of~\eqref{sys} is a nonnegative concentration vector
$x^* \in \mathbb{R}_{\geq 0}^s$ at which the
right-hand sides of the
ODEs in~\eqref{sys}  vanish: $f(x^*) =0$.
We distinguish between \defword{positive steady states} $x ^* \in \mathbb{R}^s_{> 0}$ and \defword{boundary steady states}
$x^*\in {\mathbb R}_{\geq 0}^s\backslash {\mathbb R}_{>0}^s$.
A steady state $x^*$ is \defword{nondegenerate} if
${\rm Im}\left( {\rm Jac}(f) (x^*)|_{S} \right)$
is the stoichiometric subspace $\St$.
(Here, ${\rm Jac}(f)(x^*)$ is the Jacobian matrix of $f$, with respect to $x$, at $x^*$.)
A nondegenerate steady state is
\defword{exponentially stable} if for each of the $\sigma\coloneqq \dim(\St)$ nonzero eigenvalues of ${\rm Jac}(f)(x^*)$, the real part is negative.

A network $G$ is \defword{multistationary} 
(respectively, \defword{bistable})
if, for some choice of positive rate-constant vector $\kappa \in \mathbb{R}^m_{>0}$, there exists a stoichiometric compatibility class~\eqref{eqn:invtPoly}
that contains two or more positive steady states (respectively, exponentially stable positive steady states) of~\eqref{sys}.

We analyze steady states within a stoichiometric compatibility class,
by using conservation laws
in place of linearly dependent steady-state equations, as follows.
Let $I = \{i_1 < i_2< \dots < i_d\}$ denote the set of indices of the first nonzero coordinate of the rows of the conservation-law matrix $W$.
Consider the function $\F: {\mathbb R}_{\geq 0}^s\rightarrow {\mathbb R}^s$ defined by
\begin{equation}\label{consys}
f_{c,\kappa,i} =\F(x)_i :=
\begin{cases}
f_{i}(x)&~\text{if}~i\not\in I~,\\
(Wx-c)_k &~\text{if}~i~=~i_k\in I~.
\end{cases}
\end{equation}
We call system~\eqref{consys}, the system \defword{augmented by conservation laws}. By construction, positive roots of the polynomial system $\F=0$ coincide with the positive steady states of~\eqref{sys} in the stoichiometric compatibility class~\eqref{eqn:invtPoly} defined by the total-constant vector $c$.

\subsection{Steady-state parametrizations}
The parametrizations defined below form a subclass of the ones in~\citep[Definition~3.6]{DPST} (specifically, we do not use ``effective parameters'' here).

\begin{defn}\label{def:Critical}
Let $G$ be a network with $m$ reactions, $s$ species, and
conservation-law matrix $W$.  Let $\F$ arise from $G$ and $W$ as in~\eqref{consys}.
A \defword{steady-state parametrization} is a map
 $	\phi : \mathbb{R}^{\hat m}_{>0} \times  \mathbb{R}^{\hat s}_{>0}
		\rightarrow
		 \mathbb{R}_{>0}^{m}\times  \mathbb{R}_{>0}^{s}$,
for some $\hat m \leq m$ and $\hat s \leq s$,
which we denote by
		$(\hat \kappa; \hat x) \mapsto \phi(\hat \kappa; \hat x)$,
such that:
	\begin{enumerate}[(i)]
	\item $\phi(\hat \kappa; \hat x)$ extends the vector $(\hat \kappa; \hat x)$. More precisely,
	for the natural projection $\pi: \mathbb{R}^{ m}_{>0}\times
	  \mathbb{R}_{>0}^{s}\to \mathbb{R}^{\hat m}_{>0} \times \mathbb{R}^{\hat s}_{>0} $,
	  the map $\pi \circ \phi$ is the identity map.
	\item The image of $\phi$ equals the following set:
\[
    \{ ( \kappa^*; x^*) \in \mathbb{R}^{m+s}_{>0} \mid x^* \text{ is a steady state of the system defined by } G \text{ and }  \kappa=\kappa^*\}~.
\]
 \end{enumerate}
For such a parametrization $\phi$, the \defword{critical function}
$C: \mathbb{R}^{\hat m}_{>0} \times  \mathbb{R}^{\hat s}_{>0}  \rightarrow   \mathbb{R}$ is given by:
\begin{equation*}
 C(\hat \kappa;  \hat x) \quad  = \quad \left(\det {\rm Jac}~\F \right)|_{(\kappa ; x)=\phi(\hat \kappa;  \hat x)}~,
\end{equation*}
where ${\rm Jac}(\F)$ denotes the Jacobian matrix of $\F$ with respect to $x$.
\end{defn}

The following result is implied by~\citep[Theorem 3.12]{DPST}:
\begin{proposition}[Multistationarity and critical functions] \label{prop:c-general}
Let $\phi$ be a steady-state parametrization (as in Definition~\ref{def:Critical}) for a network~$G$
that is conservative and has no boundary steady states in any compatibility class.
Let $N$ be the stoichiometric matrix of $G$.
\begin{enumerate}[(A)]
\item  {\bf Multistationarity.}
 $G$ is multistationary
if there exists
$(\hat \kappa^*;  \hat{x}^*) \in
\mathbb{R}^{\hat m}_{>0} \times  \mathbb{R}^{\hat s}_{>0} $
such that
\[
{\rm sign}( C(\hat \kappa^*;  \hat{x}^*) ) ~=~ (-1)^{\mathrm{rank}(N)+1}~.
\]
\item  {\bf Witness to multistationarity.}
Every $(\hat \kappa^*;  \hat{x}^*) \in
\mathbb{R}^{\hat m}_{>0} \times  \mathbb{R}^{\hat s}_{>0} $
with
	${\rm sign} (C(\hat \kappa^*,  x^*))=
	(-1)^{\mathrm{rank}(N)+1}$
		yields
 a witness
to multistationarity $(\kappa^*, c^*)$ as follows.
Let $(\kappa^*,x^*)=\phi(\hat \kappa^*,  \hat x^*)$.
Let $c^* = W x^*$
(so, $c^*$ is the total-constant vector defined by $x^*$,
where $W$ is the conservation-law matrix).
Then, for the mass-action system~\eqref{sys} arising from $G$ and $\kappa^*$,
there are two or more positive steady states in the stoichiometric compatibility class~\eqref{eqn:invtPoly} defined by $c^*$.
\end{enumerate}
\end{proposition}

\section{ERK networks} \label{sec:networks}
As mentioned in the Introduction, 
this work primarily concerns two networks, the minimally bistable ERK subnetwork and the reduced ERK network.  Here we recall from~\citep{OSTT} the ODEs arising from these networks and a Hopf-bifurcation criterion for the reduced ERK network (Proposition~\ref{prop:hopfcriterion}).
We also present a steady-state parametrization for the minimally bistable ERK subnetwork (Proposition~\ref{prop:param-irrev}).

\subsection{Minimally bistable ERK subnetwork}
\begin{table}[hbt]
  \centering
    \begin{tabular}[ht]{cccccccccccc} \hline
      $x_1$ & $x_2$ & $x_3$ & $x_4$ & $x_5$ & $x_6$ & $x_7$ & $x_8$ & $x_9$ & $x_{10}$ &
      $x_{11}$ & $x_{12}$
      \\ \hline
$S_{00}$ &  $E$ & $F$ & $S_{11} F$ & $S_{10} F$ & $S_{01} F$ & $S_{01} E$ & $S_{10}E$
	& $S_{01}$ & $S_{10}$  & $S_{00} E$  & $S_{11}$ \\
\hline
    \end{tabular}
  \caption{
      Assignment of variables to species for the 
      minimally bistable ERK subnetwork. }
  \label{tab:variables}
\end{table}

For the minimally bistable ERK subnetwork, let $x_1,x_2,\ldots,x_{12}$ denote the concentrations of the species in the order given in Table~\ref{tab:variables}.
We obtain the following ODE system~\eqref{sys}: 

\begin{align}\notag
   \dot{x_1} ~& =~ -k_1x_1x_2+\lcat x_5+n_3 x_6  & & ~=:~ f_1\\ \notag
     \dot{x_2} ~& =~ -k_1 x_1 x_2-\kon x_2 x_9-m_2x_{10}x_2+\kcat x_7+\koff x_7 +m_3 x_8  & & ~=:~ f_2\\ \notag
    \dot{x_3} ~& =~ -\ell_1 x_3 x_{12}-n_1 x_3 x_9+\lcat x_5+\loff x_5+n_3 x_6   & &~=:~ f_3\\ \notag
    \dot{x_4} ~& =~ \ell_1 x_3 x_{12} -\ell_3 x_4  & &~=:~ f_4\\ \notag
    \dot{x_5} ~& =~ \ell_3 x_4-\lcat x_5-\loff x_5  & &~=:~ f_5\\
    \label{eq:ODE-irre}
    \dot{x_6} ~& =~ n_1 x_3 x_9-n_3 x_6  & &~=:~ f_6\\ \notag
    \dot{x_7} ~& =~ \kon x_2 x_9+k_3 x_{11}-\kcat x_7-\koff x_7  & &~=:~ f_7\\ \notag
    \dot{x_8} ~& =~ m_2 x_2 x_{10}-m_3 x_8  & &~=:~ f_8\\ \notag
    \dot{x_9} ~& =~ -\kon x_2 x_9-n_1 x_3 x_9+\koff x_7  & &~=:~ f_9\\ \notag
    \dot{x_{10}} ~& =~ -m_2 x_2 x_{10}+\loff x_5  & &~=:~ f_{10}\\ \notag
    \dot{x_{11}} ~& =~ k_1x_1x_2 -k_3x_{11} & &~=:~ f_{11}\\ \notag
    \dot{x_{12}} ~& =~ -\ell_1 x_3 x_{12} +\kcat x_7+m_3 x_8 & &~=:~ f_{12} \notag
\end{align}

The 3 conservation equations correspond to the total amounts of substrate, kinase $E$, and phosphatase $F$, respectively:
\begin{align} \label{eq:cons-law-irre}
\notag
x_1+x_4+x_5+x_6+x_7+x_8+x_9+x_{10}+x_{11}+x_{12} ~&=~ S_\text{{tot}} ~=:~ c_1\\
x_2+x_7+x_8+x_{11} ~&=~ E_{\text{tot}} ~=:~ c_2\\
\notag
x_3+x_4+x_5+x_6 ~&=~ F_{\text{tot}}~=:~ c_3.
\end{align}

This network admits a steady-state parametrization (Proposition~\ref{prop:param-irrev} below).
Another parametrization for this network was given in~\citep[Section 3.2]{OSTT}, involving ``effective parameters'' (replacing, for instance, $\lcat/\kcat$ by a new parameter $a_1$).
That parametrization, however, does not give (direct) access to the rate constants $\kcat, \lcat, \koff, \loff $ involved in processivity levels.  We therefore need a new parametrization, as follows.

\begin{proposition}[Steady-state parametrization for minimally bistable ERK subnetwork]  \label{prop:param-irrev}
For the minimally bistable ERK subnetwork, with rate-constant vector denoted by $\kappa:=(k_1,
    k_3, \kcat, \kon, \koff,
    {\ell}_1,
    {\ell}_3, \lcat, 
    \loff,
    m_2, m_3,
    n_1, 
    n_3
    ) $,
a steady-state parametrization is given by:   \begin{align*} 
    \phi : \mathbb{R}^{13}_{>0}\times \mathbb{R}^{3}_{>0} ~& \to~ \mathbb{R}^{13}_{>0} \times \mathbb{R}^{12}_{>0} \\
    (\kappa ; ~x_1, x_2, x_{3} ) ~&\mapsto ~(\kappa ; ~x_1, x_2,\dots, x_{12})~, \notag
  \end{align*}
where
  \begin{align}\label{eq:param-irrev-details}
{ 
\begin{array}{ll}
x_4 ~=~ \frac{k_1\kcat (\lcat +\loff)(\kon x_2+n_1x_3)x_1x_2}{\ell_3\lcat (\kcat \kon x_2+\kcat n_1x_3+\koff n_1x_3)},\;\;\;&
x_5 ~=~ \frac{k_1\kcat (\kon x_2+n_1x_3)x_1x_2}{\lcat (\kcat \kon x_2+\kcat n_1x_3+\koff n_1x_3)}\\
x_6 ~=~ \frac{n_1k_1\koff x_1x_2x_3}{n_3(\kcat \kon x_2+\kcat n_1x_3+\koff n_1x_3)},\;\;\;
& x_7 ~=~ \frac{k_1(\kon x_2+n_1x_3)x_1x_2}{\kcat \kon x_2+\kcat n_1x_3+\koff n_1x_3}, \\
x_8 ~=~ \frac{k_1\kcat \loff(\kon x_2+n_1x_3)x_1x_2}{\lcat m_3(\kcat \kon x_2+\kcat n_1x_3+\koff n_1x_3)},\;\;\;& x_9 ~=~ \frac{k_1\koff x_1x_2}{\kcat \kon x_2+\kcat n_1x_3+\koff n_1x_3}, \\
x_{10} = \frac{k_1\kcat \loff(\kon x_2+n_1x_3)x_1}{\lcat m_2(\kcat \kon x_2+\kcat n_1x_3+\koff n_1x_3)}, \;\;\;& x_{11} = \frac{k_1x_1x_2}{k_3},\\
x_{12} = \frac{k_1\kcat (\lcat +\loff)(\kon x_2+n_1x_3)x_1x_2}{\lcat \ell_1(\kcat \kon x_2+\kcat n_1x_3+\koff n_1x_3)x_3} . &
\end{array}
}
\end{align}
\end{proposition}

\begin{proof}
Due to the conservation laws~\eqref{eq:cons-law-irre}, it suffices to show that by solving the equations $f_i=0$  from~\eqref{eq:ODE-irre}, for all $i \neq 2, 3, 12$, we obtain the expressions in~\eqref{eq:param-irrev-details}.
We accomplish this as follows. By solving for $x_{11}$ in the equation $f_{11}=0$,
we obtain the desired expression for $x_{11}$.
Next, we solve for $x_7$ and $x_9$ in $f_7=f_9=0$, and use the expression for $x_{11}$,
plus the fact that each $x_i$ and each rate constant is positive, to obtain
the expressions for $x_7$ and $x_9$.
Our remaining steps proceed similarly: we use $f_6=0$ to obtain $x_6$,
then $f_1=0$ for $x_5$,
then $f_{10}=0$ for $x_{10}$,
then $f_8=0$ for $x_8$,
then $f_5=0$ for $x_4$, and finally
 $f_4=0$ for $x_{12}$.
\end{proof}

\subsection{Reduced ERK network}

\begin{table}[hbt]
  \centering
    \begin{tabular}[ht]{cccccccccc} \hline
      $x_1$ & $x_2$ & $x_3$ & $x_4$ & $x_5$ & $x_6$ & $x_7$ & $x_8$ & $x_9$ & $x_{10}$ \\ \hline
$S_{00}$ &  $E$ & $S_{00} E$ & $S_{01} E$ & $S_{11}$ & $S_{01}$ & $S_{10}$ & $F$ & $S_{11} F$ & $S_{10} F$\\
\hline
    \end{tabular}
  \caption{
      Assignment of variables to species for the reduced ERK network~in Figure~\ref{fig:red_erk_net}.
}
  \label{tab:variables-reduced}
\end{table}

The reduced ERK network has 10 rate constants:
$k_1, k_3, \kcat, \koff, m, \ell_1, \ell_3, \lcat, \loff, n$.
Letting $x_1,x_2, \ldots, x_{10}$ denote the species
    concentrations in the order given in
    Table~\ref{tab:variables-reduced},
the resulting mass-action kinetics ODEs are as follows:
\begin{align}
    \notag
\dot{x_1} ~&=~ -k_1 x_{1} x_{2}+n x_{6} x_{8}+\lcat x_{10} & &~=:~ f_1 \\
    \notag
\dot{x_2} ~&=~ -k_1 x_{1} x_{2}+\kcat x_{4}+\koff x_{4} & &~=:~ f_2 \\
    \notag
\dot{x_3} ~&=~ k_1 x_{1} x_{2}-k_3 x_{3} & &~=:~ f_3
\\
    \notag
\dot{x_4} ~&=~ k_3 x_{3}-\kcat x_{4}-\koff x_{4} & &~=:~ f_4 \\
\label{eq:ODE-reduced}
\dot{x_5} ~&=~ m x_{2} x_{7}-\ell_1 x_{5} x_{8}+\kcat x_{4} & &~=:~ f_5
\\
    \notag
\dot{x_6} ~&=~ -n x_{6} x_{8}+\koff x_{4} & &~=:~ f_6 \\
    \notag
\dot{x_7} ~&=~ -m x_{2} x_{7}+\loff x_{10} & &~=:~ f_7 \\
    \notag
\dot{x_8} ~&=~ -\ell_1 x_{5} x_{8}+\loff x_{10}+\lcat x_{10} & &~=:~ f_8 \\
    \notag
\dot{x_9} ~&=~ \ell_1 x_{5} x_{8}-\ell_3 x_{9} & &~=:~ f_9 \\
    \notag
\dot{x_{10}} ~&=~ -\loff x_{10}+\ell_3 x_{9}-\lcat x_{10} & &~=:~ f_{10}.
\end{align}


\subsection{Hopf-bifurcation criterion for the reduced ERK network}
At a \defword{simple Hopf bifurcation}, a single complex-conjugate pair of eigenvalues of the Jacobian matrix crosses the imaginary axis at nonzero speed, while all other eigenvalues remain with negative real parts.
{\color{black} Such a bifurcation generates
\defword{oscillations} or periodic orbits (see, e.g., the book of~\cite{bifurcation-book})}.

\begin{defn} \label{def:hurwitz}
The $i$-th \defword{Hurwitz matrix} of a univariate polynomial
$p(\lambda)= b_0 \lambda^n + b_{1} \lambda^{n-1} + \cdots + b_n$
is the following $i \times i$ matrix:
\[
H_i ~=~
\begin{pmatrix}
b_1 & b_0 & 0 & 0 & 0 & \cdots & 0 \\
b_3 & b_2 & b_1 & b_0 & 0 & \cdots & 0 \\
\vdots & \vdots & \vdots &\vdots & \vdots &  & \vdots \\
b_{2i-1} & b_{2i-2} & b_{2i-3} & b_{2i-4} &b_{2i-5} &\cdots & b_i
\end{pmatrix}~,
\]
where the $(k,l)$-th entry is $b_{2k-l}$
  as long as
  $0 \leq 2 k - l \leq 2k-l $, and
  $0$ otherwise.
\end{defn}

The following result is~\citep[Proposition 4.1]{OSTT}.
\begin{proposition}[Hopf criterion for reduced ERK] \label{prop:hopfcriterion}
Consider the reduced ERK network,
and let $f_1,f_2,\dots, f_{10}$
denote the
right-hand sides of the resulting ODEs, as in~\eqref{eq:ODE-reduced}.
Let $\hat{\kappa}:=(\kcat, \koff, \loff)$ and $x:=(x_1, x_2, \dots, x_{10})$.
Consider the map\footnote{The map $\phi$ is a steady-state parametrization~\citep{OSTT}.}
 $\phi : \mathbb{R}^{3+10}_{>0}  \to \mathbb{R}^{10+10}_{>0}$,
denoted by $(\kcat, \koff, \loff,~x_1, x_2,\dots, x_{10}) \mapsto
( k_1, k_3, \kcat,\koff, m ,\ell_1,\ell_3,\ell_{\rm cat},\loff, n ,~x_1, x_2,\dots, x_{10}),$
 where
 \begin{align} \notag 
k_1 ~&:=~ \dfrac{(\kcat+\koff) x_4}{x_1x_2} & \quad
k_3 ~&:=~ \dfrac{(\kcat+\koff)x_4}{x_3} & \quad
m ~&:=~ \dfrac{\loff x_{10}}{x_2x_7} & \quad
\ell_1 ~&:=~ \dfrac{\loff x_{10}+\kcat x_4}{x_5x_8}\\
\notag
\ell_3 ~&:=~ \dfrac{\loff x_{10}+\kcat x_4}{x_9}& \quad
\lcat ~&:=~ \dfrac{\kcat x_4}{x_{10}} & \quad
n ~&:=~ \dfrac{\koff x_4}{x_6x_8}.
    \end{align}
Then the following is a univariate, degree-7 polynomial in $\lambda$, with coefficients in $\mathbb{Q}(x)[\hat{\kappa}]$:
 \begin{align} \label{eq:reduced-char-poly}
    q(\lambda) ~:=~
        \frac{1}{\lambda^3} ~ \det \left( \lambda I - {\rm Jac} (f) \right) |_{(\kappa; x)= \phi(\hat{\kappa}; x)} ~.
 \end{align}
 Now let $\mathfrak{h}_i$, for $i=4,5,6$, denote the determinant of the $i$-th Hurwitz matrix of the polynomial~$q(\lambda)$ in~\eqref{eq:reduced-char-poly}.
Then the following are equivalent:
\begin{enumerate}[(i)]
    \item there exists a rate-constant vector $\kappa^* \in \mathbb{R}^{10}_{>0}$ such that the resulting system~\eqref{eq:ODE-reduced} exhibits a simple Hopf bifurcation, with respect to $\kcat$, at some
    $x^* \in \mathbb{R}^{10}_{>0}$, and
    \item there exist  $x^* \in \mathbb{R}^{10}_{>0}$ and $\hat{\kappa}^*\in \mathbb{R}^3_{>0}$ such that
    \begin{align} \label{eq:hurwitz-hopf-conditions-reduced-ERK}
   \mathfrak{h}_4 (\hat{\kappa}^*; x^*) >&0~, ~
    \mathfrak{h}_5 (\hat{\kappa}^*; x^*) >0~, ~
    \mathfrak{h}_6 (\hat{\kappa}^*; x^*) =0~, ~ 
    \frac{\partial}{\partial \kcat } 
    \mathfrak{h}_6 (\hat{\kappa}; x) |_{(\hat{\kappa}; x)=(\hat{\kappa}^*; x^*)} \neq 0~.
    \end{align}
\end{enumerate}
Moreover, given $\hat{\kappa}^*$ and $x^*$ as in {\color{blue} (ii)}, a simple Hopf bifurcation with respect to $\kcat$ occurs at~$x^*$ when the rate-constant
vector is $\kappa^*:= \widetilde{\pi} (\phi(\hat{\kappa}^*; x^*))$.
Here, $\widetilde{\pi}:  \mathbb{R}_{>0}^{10}\times  \mathbb{R}_{>0}^{10} \to \mathbb{R}_{>0}^{10}$ is the natural projection to the first 10 coordinates.
 \end{proposition}

\section{Bistability} \label{sec:bistable}
In this section, we show that, for {\em every} choice of processivity levels, the minimally bistable ERK network is multistationary  (Theorem~\ref{thm:MSS}).  We also give 
evidence suggesting that in fact, when we have multistationarity, we always have bistability (Section~\ref{sec:bistab-evidence}).
Finally, we investigate multistationarity numerically for processivity levels close to 1 (Section~\ref{sec:numerics}).

\subsection{Multistationarity at all processivity levels} \label{sec:mss}

\begin{theorem}[Multistationarity at all processivity levels] \label{thm:MSS}
Consider the minimally bistable ERK subnetwork.
For every choice of processivity levels
$p_k \in (0,1)$ and  $p_{\ell} \in (0,1)$,
there is a rate-constant vector
$(k_1^*,
    k_3^*, \kcat^*, \kon^*, \koff^*,
    {\ell}_1^*,
    {\ell}_3^*, \lcat^*, 
    \loff^*,
    m_2^*, m_3^*,
    n_1^*, 
    n_3^*
    ) \in \mathbb{R}^{13}_{>0}$
such that
    \begin{enumerate}
        \item $p_k={\kcat^*}/{(\kcat^* + \koff^*)}$ and $p_{\ell}={\lcat^*}/{(\lcat^*+\loff^*)}$, and
        \item the resulting system admits multiple positive steady states (in some compatibility class).
    \end{enumerate}
\end{theorem}

\begin{proof}
Let $C( \kappa; \hat x)$ (where $\hat x=(x_1, x_2, x_3)$) denote the critical function of the steady-state parametrization \eqref{eq:param-irrev-details} in Proposition~\ref{prop:param-irrev}.

By setting $\koff^*=\loff^*=1$ and allowing $\kcat^*$ and $\lcat^*$ to be arbitrary positive values, we obtain all processivity levels $p_k={\kcat^*}/{(\kcat^*+\koff^*)}$ and $p_{\ell}={\lcat^*}/{(\lcat^*+\loff^*)}$ in $(0,1)$.
Also, the rank of stoichiometric matrix $N$ for this network is 9; hence, $(-1)^{\mathrm{rank}(N)+1} = 1$.
So, by Proposition~\ref{prop:c-general}, it suffices to show that for all $\kcat^*>0$ and $\lcat^*>0$, the following specialization of the critical function
is positive when we further specialize at some choice of
$(k_1,
    k_3, \kon,
    {\ell}_1,
    {\ell}_3,
    m_2, m_3,
    n_1, 
    n_3
    )  \in \mathbb{R}^9_{>0}$,
and $\hat x \in \mathbb{R}^3_{>0}$:
\begin{align} \label{eq:C-specialized}
        C( \kappa; \hat x) |_{\koff=\loff=1,~\kcat=\kcat^*,~\lcat=\lcat^*}
\end{align}

To see that the function~\eqref{eq:C-specialized} can be positive,
first note that the denominator of
$   C( \kappa; \hat x) |_{\koff=\loff=1}$, shown here, is always positive (all rate constants and $x_i$'s are positive):
\[
        (\kcat  \kon  x_2 + \kcat  n_1  x_3 +n_1  x_3 )^2 \lcat  x_3~.
\]
(See the supplementary file {\tt minERK-mss-bistab.mw}.)
Thus, it suffices to analyze the numerator of
$    C( \kappa; \hat x) |_{\koff=\loff=1}$. We denote this numerator by $\widetilde{C}$, and specialize as follows to obtain (see the supplementary file):
\begin{align} \label{eq:specialized-c}
 &
\widetilde{C}|_{
k_1=t^{-1},
    k_3=t^{-1},
    \kon=1,
    {\ell}_1=t,
    {\ell}_3=t^{-1},
    m_2=1,
    m_3=1,
    n_1=1,
    n_3=1,
    x_1=t, x_2=t, x_3=1
    } \\
    & \quad
~=~
   (2 \kcat^2 \lcat^2+2 \kcat^2 \lcat) t^5 \label{eq:specialized-c-leading-coeff}
   \\
    \notag
   & \quad \quad \quad
   +(-4 \kcat^3 \lcat^2-3 \kcat^3 \lcat+3 \kcat^2 \lcat^2-\kcat^3+9 \kcat^2 \lcat+3 \kcat \lcat^2+2 \kcat^2+3 \kcat \lcat) t^4  \\  \notag
      & \quad \quad \quad
   + \text{lower-order terms in } t.
\end{align}
Therefore, for all $\kcat>0$ and $\lcat>0$, the leading coefficient with respect to $t$ in~\eqref{eq:specialized-c-leading-coeff} is positive and so
the specialization of $\widetilde{C}$ is positive for sufficiently large $t$, which yields the desired values for the rate constants shown in~\eqref{eq:specialized-c}.
\end{proof}

\begin{remark} \label{rmk:NP}
In the proof of Theorem~\ref{thm:MSS}, the specialization~\eqref{eq:specialized-c} was obtained by viewing $\widetilde{C}$ as a polynomial 
in which each coefficient is a polynomial in $\kon$, $\kcat$ and $\lcat$, and then analyzing the resulting Newton polytope in a standard way (cf.\ \citep[Lemma B.3]{OSTT}), as follows.  We first found a vertex of the polytope whose corresponding coefficient is a positive polynomial (namely, the leading coefficient in~\eqref{eq:specialized-c-leading-coeff}).
Next, we chose a vector $v$ (specifically, $v = [1, 1, 0, -1, -1, 1, -1, 0, 0, 0, 0]$) in the interior of the corresponding cone in the polytope's outer normal fan.
So, by substituting $\kon=1$ and $t^{v_1}, t^{v_2},\dots $ for the variables $x_1, x_2, x_3, k_1, k_3, \ell_1, \ell_3, m_2, m_3, n_1, n_3$, the resulting polynomial is positive for large $t$.
\end{remark}

\subsection{Evidence for bistability} \label{sec:bistab-evidence}
Theorem~\ref{thm:MSS} states that the minimally bistable ERK network is multistationary at all processivity levels.  Multistationarity is a necessary condition for bistability, which is the focus of the original Question~\ref{q:subnet-bistable} from the Introduction.
Accordingly, we show bistability at many processivity levels with $p_k=p_{\ell}$ (Proposition~\ref{prop:bistab-equalprocesslevels}).
Furthermore, we give additional evidence for bistability at {\em all} processivity levels (Remark~\ref{rmk:bistab-random}), which we state as Conjecture~\ref{conj:bistab}.

\begin{remark}[Assessing bistability is difficult] \label{rmk:difficulty-bistab}
Although there are many criteria for checking whether a network is multistationary, there are relatively few for checking bistability~\citep{torres-feliu}.  Moreover, here we consider a more difficult question: does our network exhibit bistability for an infinite family of parameters (rather than a single parameter vector), encompassing all processivity levels?
Thus, it is perhaps unsurprising that we obtain only partial results in this direction.  Another ``infinite'' analysis of bistability was performed recently by \cite{tang-wang}, 
who proved that an infinite family of sequestration networks all are bistable. 
\end{remark}

\begin{proposition}[Bistability at many processivity levels] \label{prop:bistab-equalprocesslevels}
Consider the minimally bistable ERK subnetwork.
For each of the following processivity levels:
\begin{align} \label{eq:specific-values}
p_k ~=~ p_{\ell} ~\in~ \{0.1, 0.2, \dots, 0.9,~ 0.91, 0.92, \ldots, 0.99\}~,
\end{align}
there is a rate-constant vector
$(k_1^*,
    k_3^*, \kcat^*, \kon^*, \koff^*,
    {\ell}_1^*,
    {\ell}_3^*, \lcat^*, 
    \loff^*,
    m_2^*, m_3^*,
    n_1^*, 
    n_3^*
    ) \in \mathbb{R}^{13}_{>0}$
such that
        $p_k={\kcat^*}/{(\kcat^* + \koff^*)}$ and $p_{\ell}={\lcat^*}/{(\lcat^*+\loff^*)}$, and
        the resulting system admits multiple \underline{exponentially stable} positive steady states (in some compatibility class).
\end{proposition}
\begin{proof}
As in the proof of Theorem~\ref{thm:MSS},
we achieve
each value of $p_k^*=p_{\ell}^*$, as in~\eqref{eq:specific-values},
by setting $\koff^* = \loff^* =1$ and $\kcat^* = \lcat^* = p_k^*/(1-p_k^*)$.

Next, we follow the proof of Theorem~\ref{thm:MSS} to find a witness to multistationarity.
Recall that the specialized numerator of the critical function given in~\eqref{eq:specialized-c}, which is a polynomial in $\kcat$, $\lcat$, and $t$, is positive (indicating multistationarity) for sufficiently large $t$. That is, there exists
a $T\in {\mathbb R}_{>0}$, which depends on the value of $p_k^*=p_{\ell}^*$,
at which the specialized critical function is positive for all $t\geq T$. For each value of $p_k^*=p_{\ell}^*$, we pick such a positive number $T$, as follows:
\begin{center}
\begin{tabular}{c|ccccccccc}
\hline
$p_k^*=p_{\ell}^*$ & 0.1 & 0.2 & 0.3 & 0.4 & 0.5 & 0.6 & 0.7 & 0.8 & 0.9 \\
$T$ & 3 & 3 & 3 & 3 & 4 & 5 & 7 & 10 & 20\\ \hline \hline
$p_k^*=p_{\ell}^*$ & 0.91 & 0.92 & 0.93 & 0.94 & 0.95 & 0.96 & 0.97 & 0.98 & 0.99\\
$T$ & 22 & 25 & 28 & 33 & 40 & 50 & 66 & 100 & 200\\
\hline
\end{tabular}
\end{center}

It follows, from~\eqref{eq:specialized-c} and Proposition~\ref{prop:c-general}(B), that with the following rate-constant vector:
\begin{align}
    \notag
    \kappa^*
   ~:=~ &
   (k_1^*,
    k_3^*, \kcat^*, \kon^*, \koff^*,
    {\ell}_1^*,
    {\ell}_3^*, \lcat^*, 
    \loff^*,
    m_2^*, m_3^*,
    n_1^*, 
    n_3^*
    ) \\
    ~=~
    &
    (
    T^{-1},~
    T^{-1}, ~
    p_k^*/(1-p_k^*),~
    1,  ~
    1,~
    T,~
    T^{-1},~
        p_k^*/(1-p_k^*),~
    1, ~
    1, ~
    1,~
    1, ~
    1
    )~, \label{eq:special-rate-constants}
\end{align}
there are multiple steady states in the compatibility class containing
$x^*:=\pi(\phi(\kappa^*; 1,T,1))$, 
where $\phi: \mathbb{R}^{13}_{>0}\times \mathbb{R}^{3}_{>0} ~\to~ \mathbb{R}^{13}_{>0} \times \mathbb{R}^{12}_{>0}$ is the steady-state parametrization in Proposition~\ref{prop:param-irrev} and $\pi:\mathbb{R}^{13}_{>0} \times \mathbb{R}^{12}_{>0} \to \mathbb{R}^{12}_{>0}$ denotes the canonical projection to the last 12 coordinates.

Finally, for each such $x^*$ (one for each choice of $p_k^*=p_{\ell}^*$),
the stoichiometric compatibility class
of $x^*$
contains exactly three positive steady states (arising from the rate-constant vector $\kappa^*$);
see {\tt minERK-mss-bistab.mw}.
Moreover, two of the steady states each have
three zero eigenvalues and the remaining eigenvalues having strictly negative real parts (indicating that these two steady states are exponentially stable), and one steady state has a (single) non-zero eigenvalue with positive real part (indicating it is unstable); see the supplementary file. 
Therefore, we have bistability for each of the processivity levels
in~\eqref{eq:specific-values}.
%
\end{proof}


%

Proposition~\ref{prop:bistab-equalprocesslevels} showed bistability for certain processivity levels with $p_k=p_{\ell}$.  Even when $p_k \neq p_{\ell}$ (see Remark~\ref{rmk:bistab-random}), we found -- in every instance we examined -- bistability.
\begin{remark}[Bistability at random processivity levels] \label{rmk:bistab-random}
For the minimally bistable ERK subnetwork,
we generated random pairs of processivity levels $p_k$ and $p_{\ell}$ between 0 and 1 (Table~\ref{tab:bistab-ss-evals}).
For all such pairs, following the procedure described in the proof of Proposition~\ref{prop:bistab-equalprocesslevels}, we found bistability.  For details, see the supplementary file {\tt minERK-MSS-bistab.mw}.
\begin{table}[ht!]
\centering
\begin{tabular}{cccccccccc}

\hline
$p_k$ & {0.01570} & {0.02229} & {0.06748} & {0.2203} & {0.2268} & {0.2576} & {0.2897} & {0.4613} & {0.5378}\\
$p_{\ell}$ & {0.05004} & {0.3476} & {0.6011} & {0.6076} & {0.9461} & {0.2263} & {0.9883} & {0.4217} & {0.3770}\\
\hline \hline
$p_k$ & {0.5893} & {0.6613} & {0.6968} & {0.9076} & {0.9307} & {0.9598} & {0.9771} & {0.9845} \\
$p_{\ell}$ & {0.5289} & {0.04355} & {0.1351} & {0.2668} & {0.9010} & {0.6118} & {0.07128} & {0.9809}\\
\hline
\end{tabular}




















\caption{
Randomly generated pairs of processivity levels, rounded to four significant digits.  At every such pair, the minimally bistable ERK network exhibits bistability (in some compatibility class).
Computations are in the supplementary file {\tt minERK-MSS-bistab.mw}.
}
\label{tab:bistab-ss-evals}
\end{table}


\end{remark}

In light of Proposition~\ref{prop:bistab-equalprocesslevels} and
Remark~\ref{rmk:bistab-random}, we conjecture that, in Theorem~\ref{thm:MSS}, multistationarity can be strengthened to bistability.  In other words, we conjecture that the answer to Question~\ref{q:subnet-bistable} is ``yes'':

\begin{conjecture}[Bistability at all processivity levels] \label{conj:bistab}
Consider the minimally bistable ERK subnetwork.
For every choice of processivity levels
$p_k \in (0,1)$ and  $p_{\ell} \in (0,1)$,
there is a rate-constant vector
$(k_1^*,
    k_3^*, \kcat^*, \kon^*, \koff^*,
    {\ell}_1^*,
    {\ell}_3^*, \lcat^*, 
    \loff^*,
    m_2^*, m_3^*,
    n_1^*, 
    n_3^*
    ) \in \mathbb{R}^{13}_{>0}$
such that
        $p_k={\kcat^*}/{(\kcat^* + \koff^*)}$ and $p_{\ell}={\lcat^*}/{(\lcat^*+\loff^*)}$, and
        the resulting system admits multiple \underline{exponentially stable} positive steady states (in some compatibility class).
\end{conjecture}

If Conjecture~\ref{conj:bistab} holds, then \citep[Theorem 3.1]{Joshi-Shiu-2013} implies that bistability ``lifts'' to the original ERK network.  In other words, this would answer in the affirmative the original Question~\ref{q:main}, for bistability.

\subsection{Numerical investigation
for
processivity levels near 1}
\label{sec:numerics}

In this subsection, we numerically investigate multistationarity of the minimally bistable ERK network, for processivity levels close to 1.
Specifically, we examine how processivity levels near 1 affect the S-shaped steady-state curves
(as in~\citep[Figure 9.6]{distefano}) usually associated with multistationarity.  We focus in particular on the concentration of the fully phosphorylated substrate ($x_{12}$), as
this species is arguably the most interesting in our signaling network.
Indeed, this substrate is generally further processed by other signaling
modules.

\begin{figure}
  \centering

  \begin{subfigure}{0.4\textwidth}
    \includegraphics[width=0.9\linewidth]{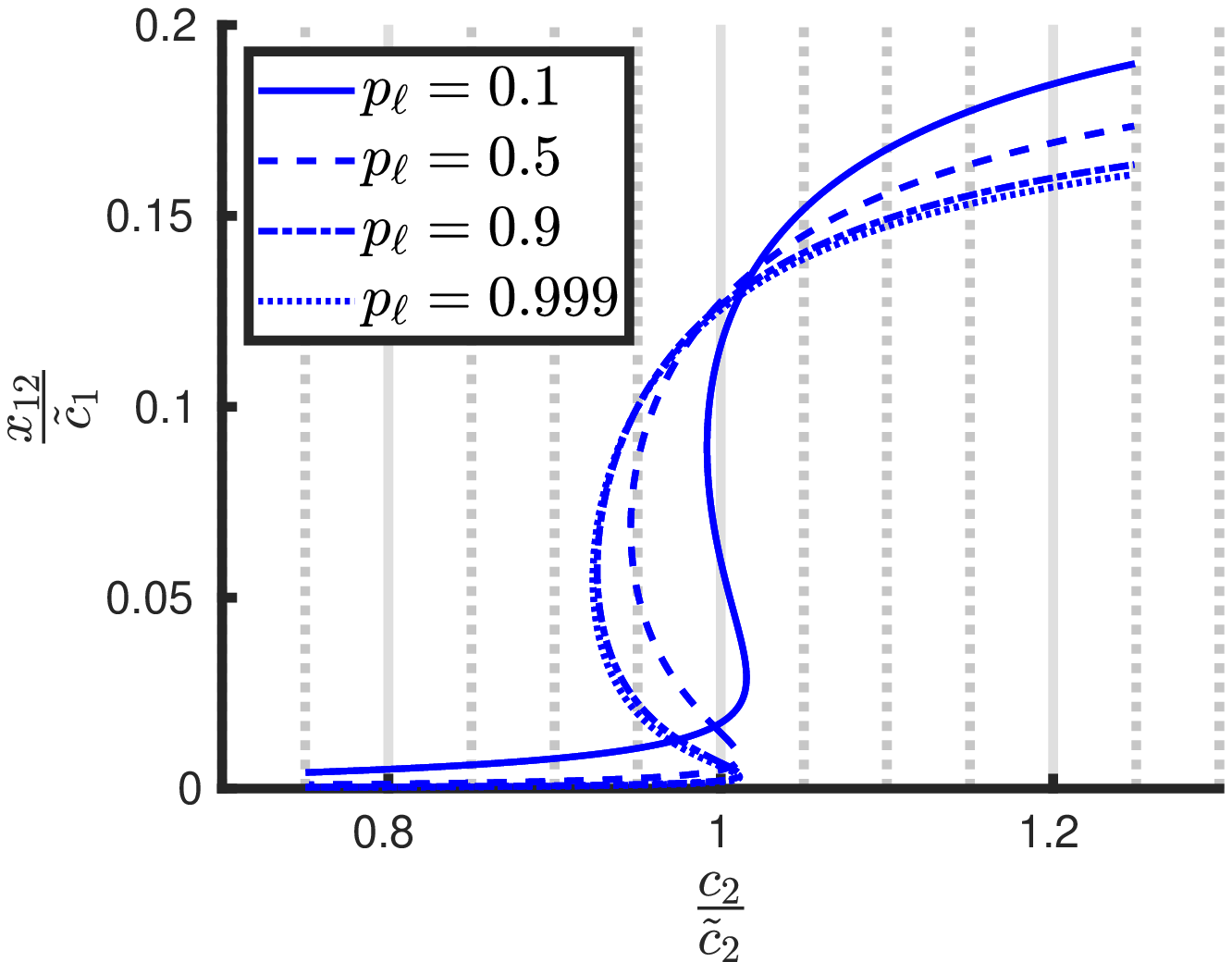}
    \subcaption{\label{fig:pl_only_full}$0.75\leq \frac{c_2}{\tilde c_2} \leq 1.25$}
  \end{subfigure}
  \hfill
  \begin{subfigure}{0.4\textwidth}
    \includegraphics[width=0.9\linewidth]{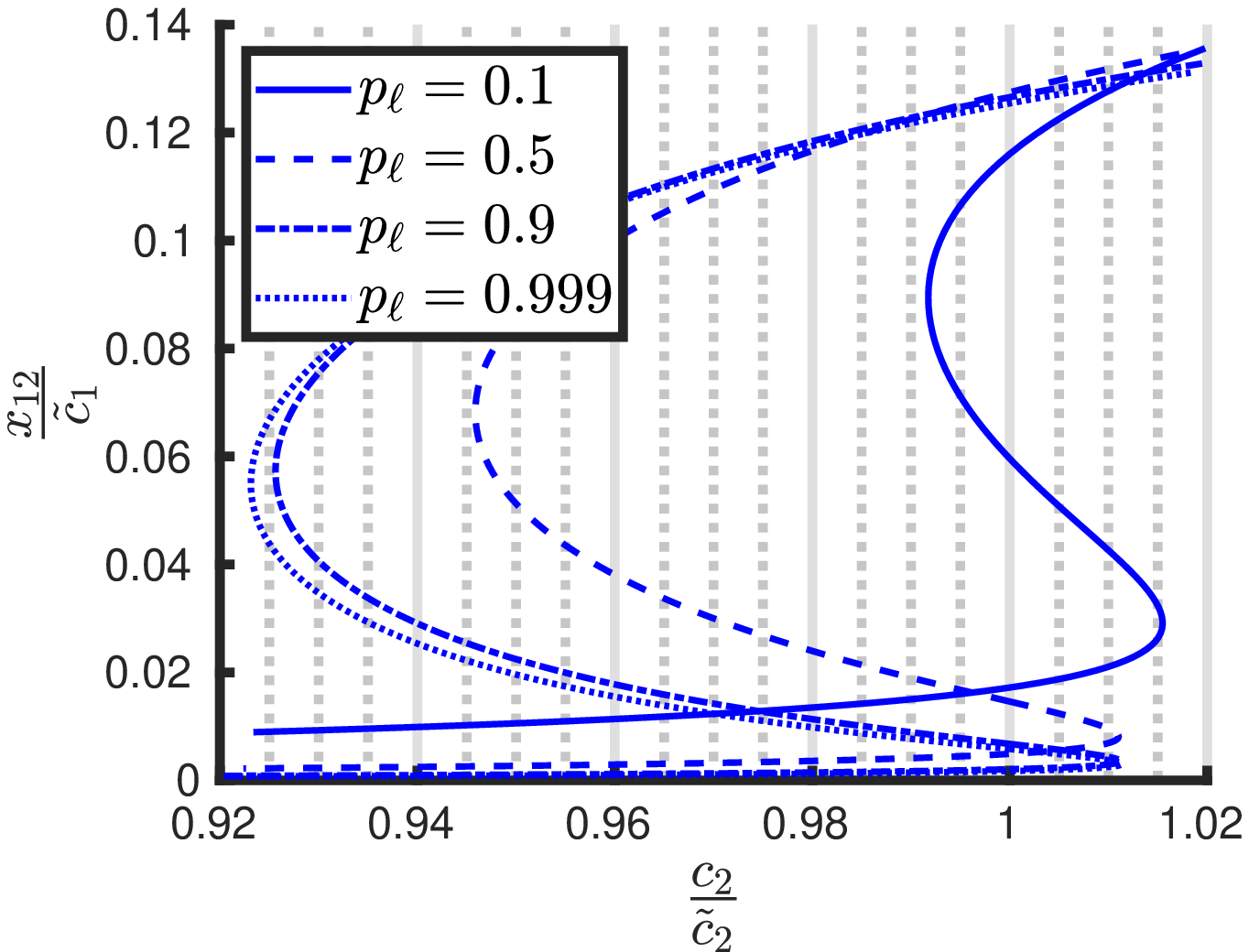}
    \subcaption{\label{fig:pl_only_zoom}$0.9\leq \frac{c_2}{\tilde c_2}\leq 1.05$}
  \end{subfigure}

  \caption{\label{fig:pl_only}
    Numerical investigation of multistationarity as $p_\ell \to 1$
    (for $p_k=0.1$;
        see Section~\ref{sec:setup-figures} and Appendix~\ref{app:num-mss} for figure setup and generation).
    An increase of $p_\ell$ leads to a (small) decrease of
    $\frac{x_{12}}{\tilde c_1}$ at $\frac{c_2}{\tilde c_2} \approx
    1.25$ (display~\subref{fig:pl_only_full}) and to a larger multistationarity interval (from approximately $0.99 \leq
    \frac{c_2}{\tilde c_2} \leq 1.01$ to $0.92 \leq \frac{c_2}{\tilde
      c_2} \leq 1.02$) (display~\subref{fig:pl_only_zoom}).
  }
\end{figure}

\subsubsection{Setup for Figures~\ref{fig:pl_only}--\ref{fig:pk_pl_099}.} \label{sec:setup-figures} Figures~\ref{fig:pl_only}--\ref{fig:pk_pl_099} were generated by
numerical continuation using {\tt Matlab} and
   {\tt Matcont}.
   Further details on how we obtained these figures are in Appendix~\ref{app:num-mss}.
   In particular, parameter values, total concentrations, and initial
    conditions were obtained by equation  (\ref{eq:special-rate-constants}) and also (in the appendix) (\ref{eq:num_C})--(\ref{eq:ss_para})
    and the values in Tables~\ref{tab:pk_pl_T_1_2}--\ref{tab:pk_pl_T_3}. In all figures, the x-axis
    is the relative total
    amount of kinase ($c_2 / \tilde c_2$ obtained in
    step~\ref{item:compute_tots} of the procedure described in the appendix), and the y-axis is the 
    relative
    amount of fully phosphorylated substrate ($x_{12} / \tilde c_1$), also obtained in step~\ref{item:compute_tots}).
    The reason for examining relative (rather than actual) amounts is that, as $p_k$ and/or $p_\ell$ approach 1,
certain total amounts differ by orders of
magnitude, and so it is more meaningful to compare
values relative to a reference point.

\subsubsection{Results} \label{sec:results-figures}
Figure~\ref{fig:pl_only} shows that,
for $p_k=0.1$ and various values of $p_\ell$, we obtain
classical S-shaped curves often associated with
multistationarity.
We also see that increasing
$p_\ell$ alone has only a modest effect on the curve:
at
the relative total concentration $\frac{c_2}{\tilde c_2} \approx
    1.25$,
the fraction of fully
phosphorylated substrate $\frac{x_{12}}{\tilde c_1}$ (at steady state) decreases but only by a small amount
(see Figure~\ref{fig:pl_only_full}).

Next, we investigate the interval of values of ${c_2}/{\tilde c_2}$ at which multistationarity occurs, which we call the \defword{multistationarity interval}.
We see in Figure~\ref{fig:pl_only_zoom}
{\color{black} (which is a ``zoomed in'' version of Figure~\ref{fig:pl_only_full})}
that as $\pl$ increases (with $p_k=0.1$), the multistationarity interval
enlarges (see the caption of Figure~\ref{fig:pl_only_zoom}).
We can view the size of this interval as a measure
of the robustness of
multistationarity with respect to 
fluctuations of the total
amount of kinase.
Hence, Figure~\ref{fig:pl_only} motivates us to
conjecture that increasing only one processivity level leads to increased robustness of multistationarity, as follows:
{\em When one processivity level is fixed and close to 0,
increasing the other
processivity level leads to a larger multistationarity interval.}

Next, we fix $p_{\ell}$ at a high value
(namely,
$p_\ell=0.9$) and increase $p_k$ (see~Figure~\ref{fig:pk_pl_09}). Again, increasing $p_k$ reduces the
fraction of fully phosphorylated substrate $\frac{x_{12}}{\tilde c_1}$ at
$\frac{c_2}{\tilde c_2}\approx 1.25$, now substantially. Moreover,
the multistationarity interval shrinks
(see Figures~\ref{fig:pk_pl_09_zoom1} and \ref{fig:pk_pl_09_zoom2}). This
motivates the following conjecture:
{\em When one processivity level is fixed and close to 1,
increasing the other
processivity level leads to a smaller multistationarity interval.}


\begin{figure}[htb]
  \centering

  \begin{subfigure}{0.3\textwidth}
    \includegraphics[width=0.9\linewidth]{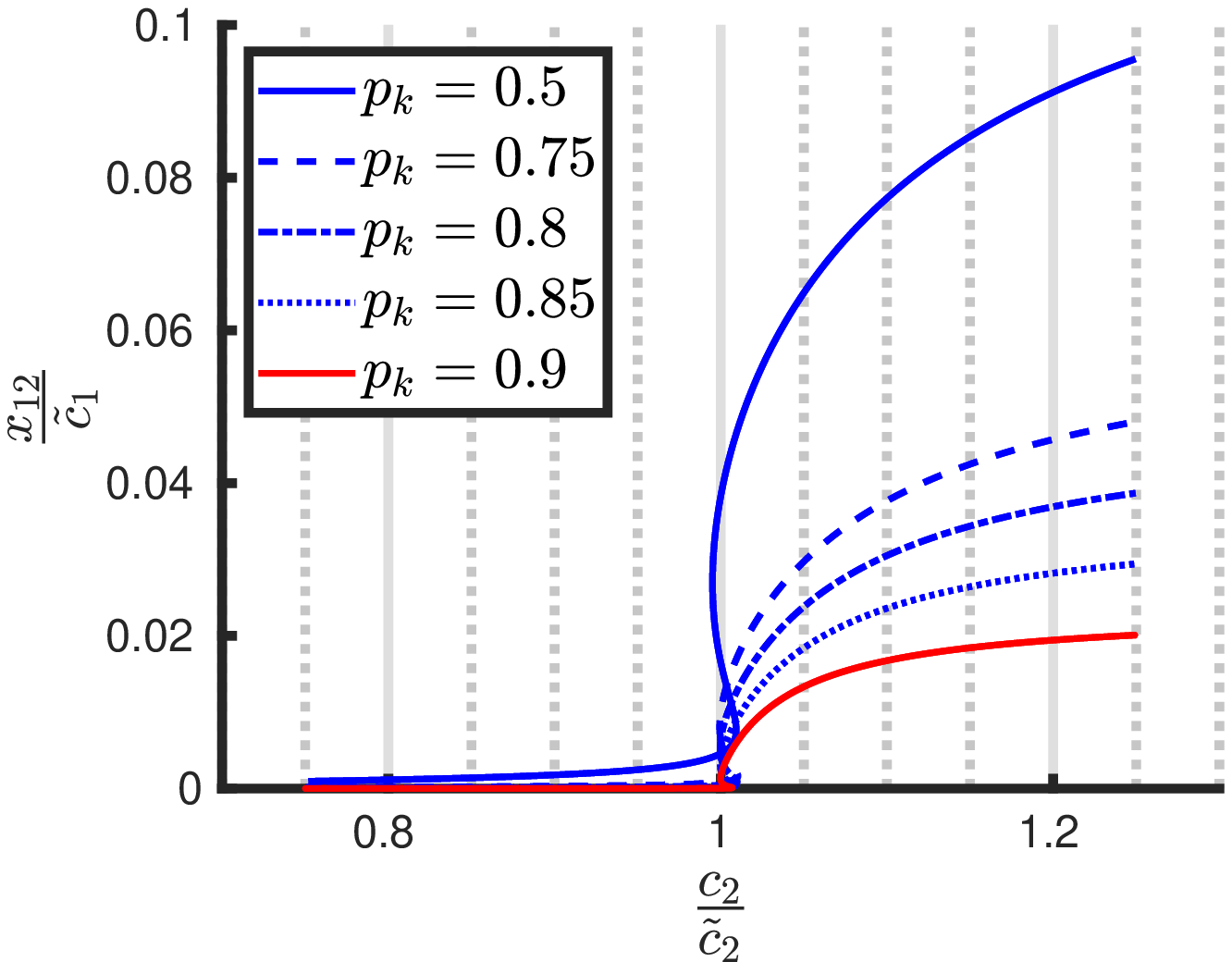}
    \subcaption{\label{fig:pk_pl_09_full}$0.75\leq \frac{c_2}{\tilde c_2} \leq 1.25$}
  \end{subfigure}
  \hfill
  \begin{subfigure}{0.3\textwidth}
    \includegraphics[width=0.9\linewidth]{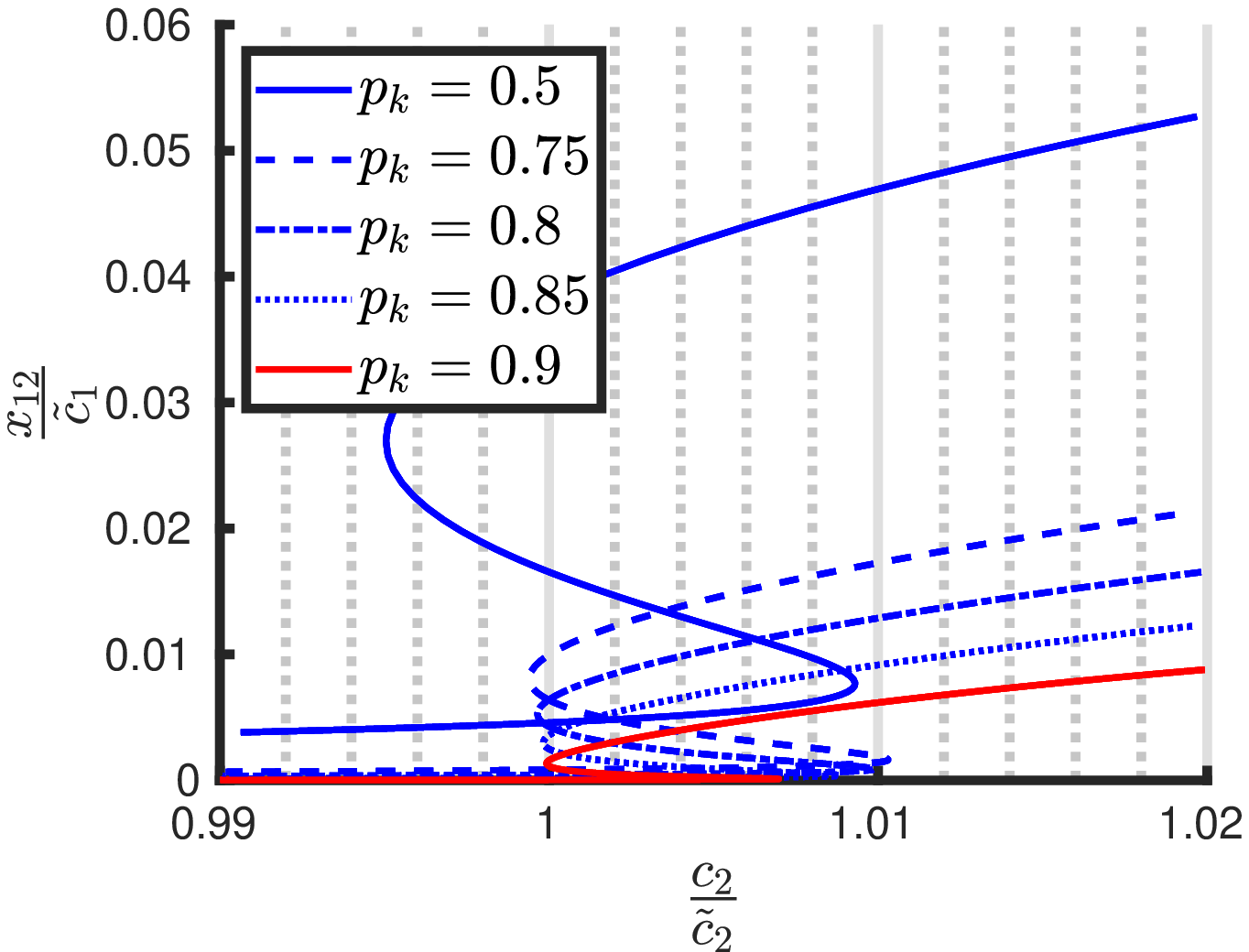}
    \subcaption{\label{fig:pk_pl_09_zoom1}$0.99\leq \frac{c_2}{\tilde c_2}\leq 1.05$}
  \end{subfigure}
  \hfill
  \begin{subfigure}{0.3\textwidth}
    \includegraphics[width=0.9\linewidth]{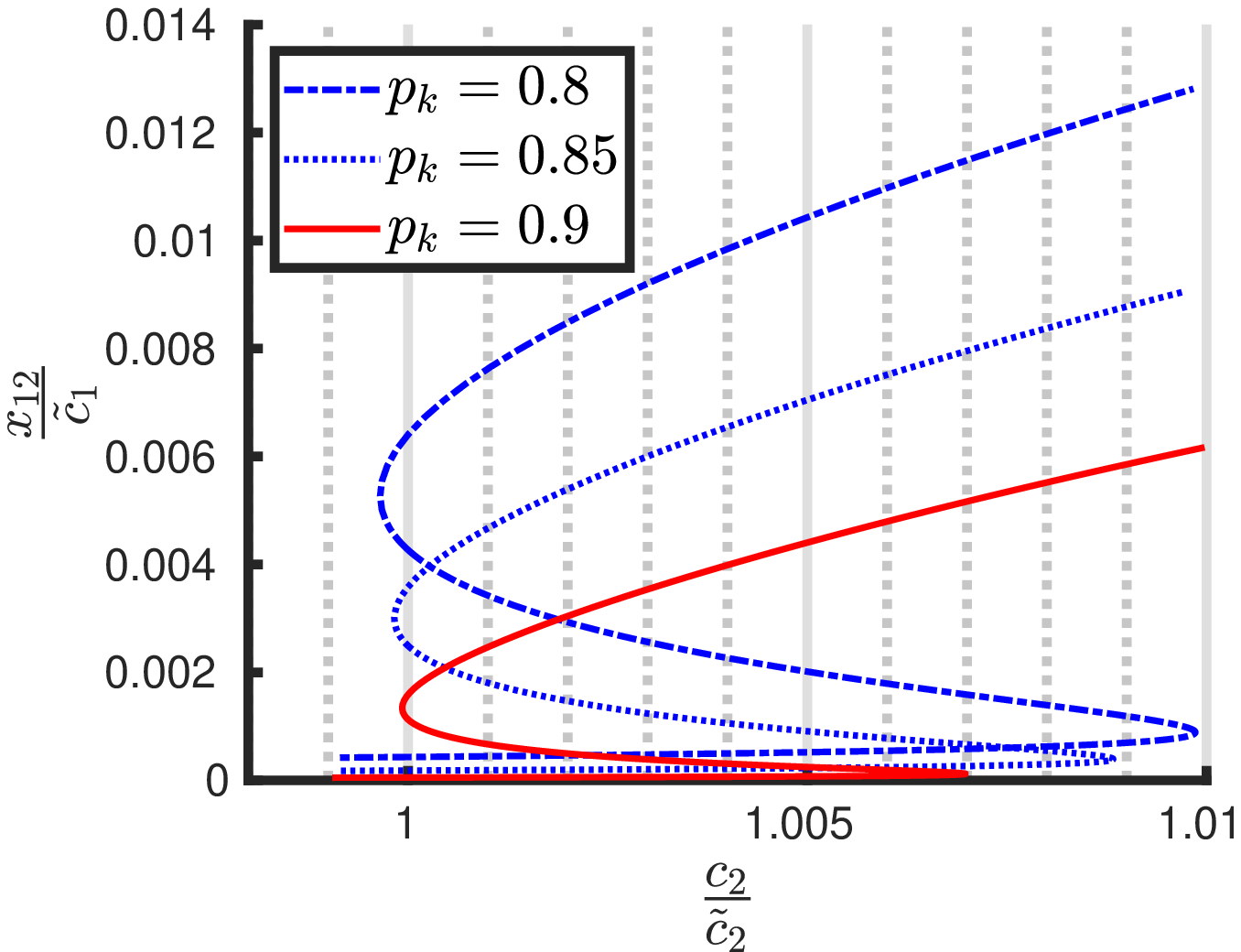}
    \subcaption{\label{fig:pk_pl_09_zoom2}$0.99\leq \frac{c_2}{\tilde c_2}\leq 1.01$}
  \end{subfigure}

  \caption{\label{fig:pk_pl_09}
    Numerical investigation of multistationarity as $p_k \to 1$
    (for $p_\ell=0.9$).
        An increase in $p_k$ leads
    to a substantial decrease in $\frac{x_{12}}{\tilde c_1}$ at
    $\frac{c_2}{\tilde c_2} \approx 1.25$
    (display~\subref{fig:pk_pl_09_full}) and a smaller
    multistationarity interval (from $0.992 \leq \frac{c_2}{\tilde
      c_2} \leq 1.01$ to $1 \leq \frac{c_2}{\tilde c_2} \leq 1.007$)
    (displays~\subref{fig:pk_pl_09_zoom1}--\subref{fig:pk_pl_09_zoom2}).
  }
\end{figure}

Finally, in Figure~\ref{fig:pk_pl_099},
we investigate values of $p_k = p_\ell$ close to 1.
Now the multistationarity interval becomes vanishingly small (see, in particular, Figure~\ref{fig:pk_pl_099_zoom22}), leading to a steady-state function that approaches a steep Hill function.
We conjecture that this phenomenon is the norm: {\em
As both processivity levels approach 1, the length of the multistationarity interval approaches 0.}


\begin{figure}[htb]
  \centering

  \begin{subfigure}{0.3\textwidth}
    \includegraphics[width=0.9\linewidth]{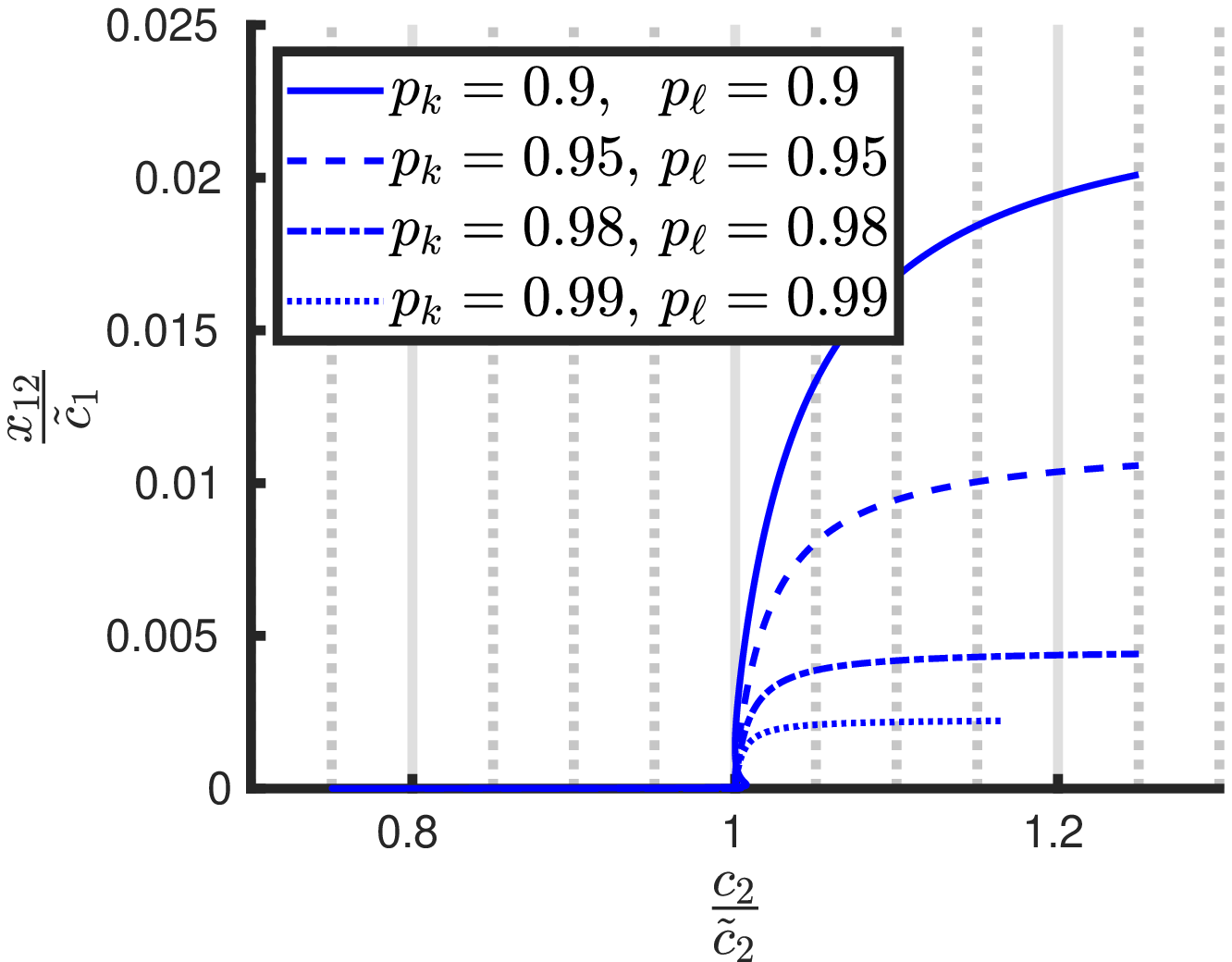}
    \subcaption{\label{fig:pk_pl_09_full2}$0.75\leq \frac{c_2}{\tilde c_2} \leq 1.25$}
  \end{subfigure}
  \hfill
  \begin{subfigure}{0.3\textwidth}
    \includegraphics[width=0.9\linewidth]{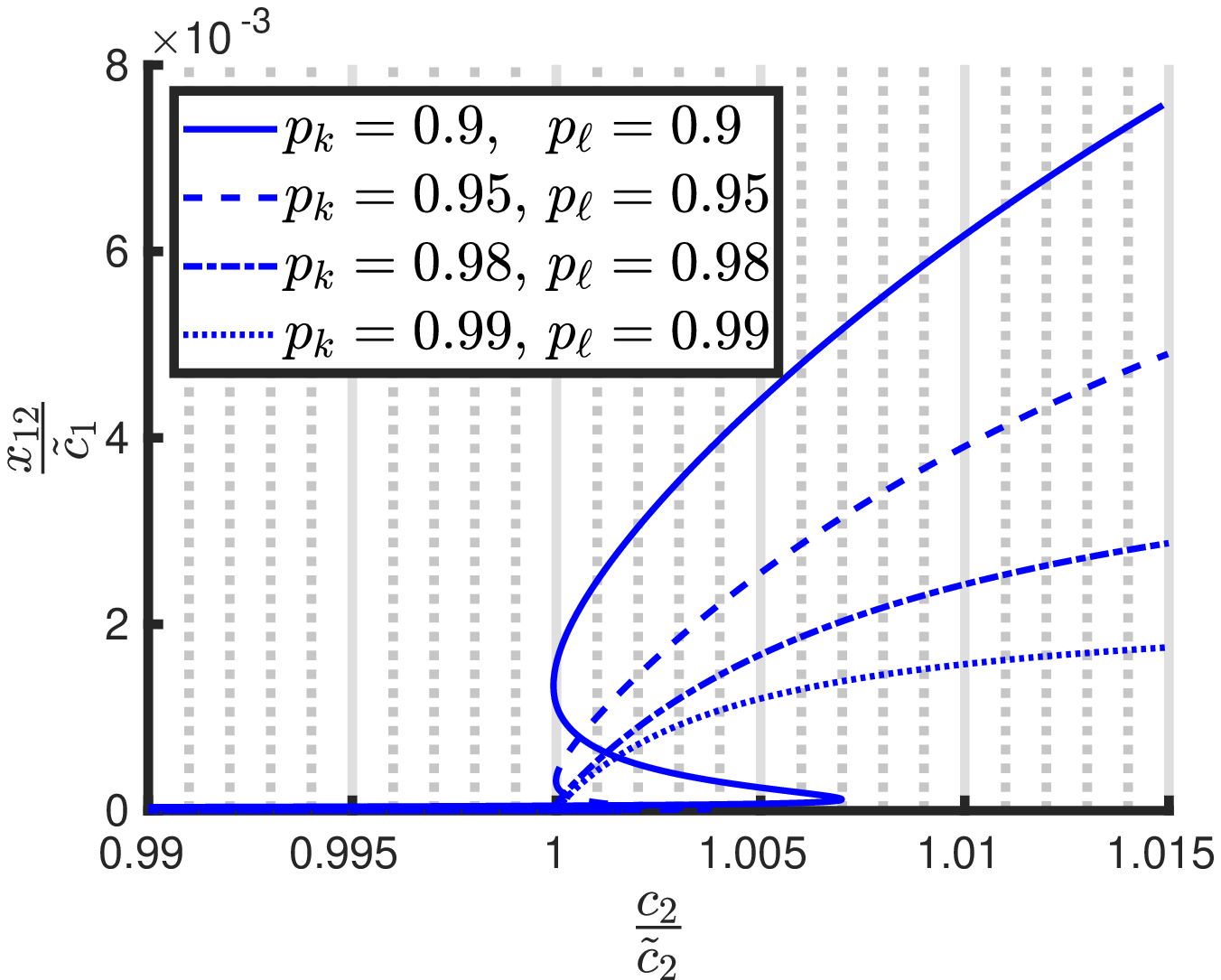}
    \subcaption{\label{fig:pk_pl_099_zoom12}$0.99\leq \frac{c_2}{\tilde c_2}\leq 1.05$}
  \end{subfigure}
  \hfill
  \begin{subfigure}{0.3\textwidth}
    \includegraphics[width=0.9\linewidth]{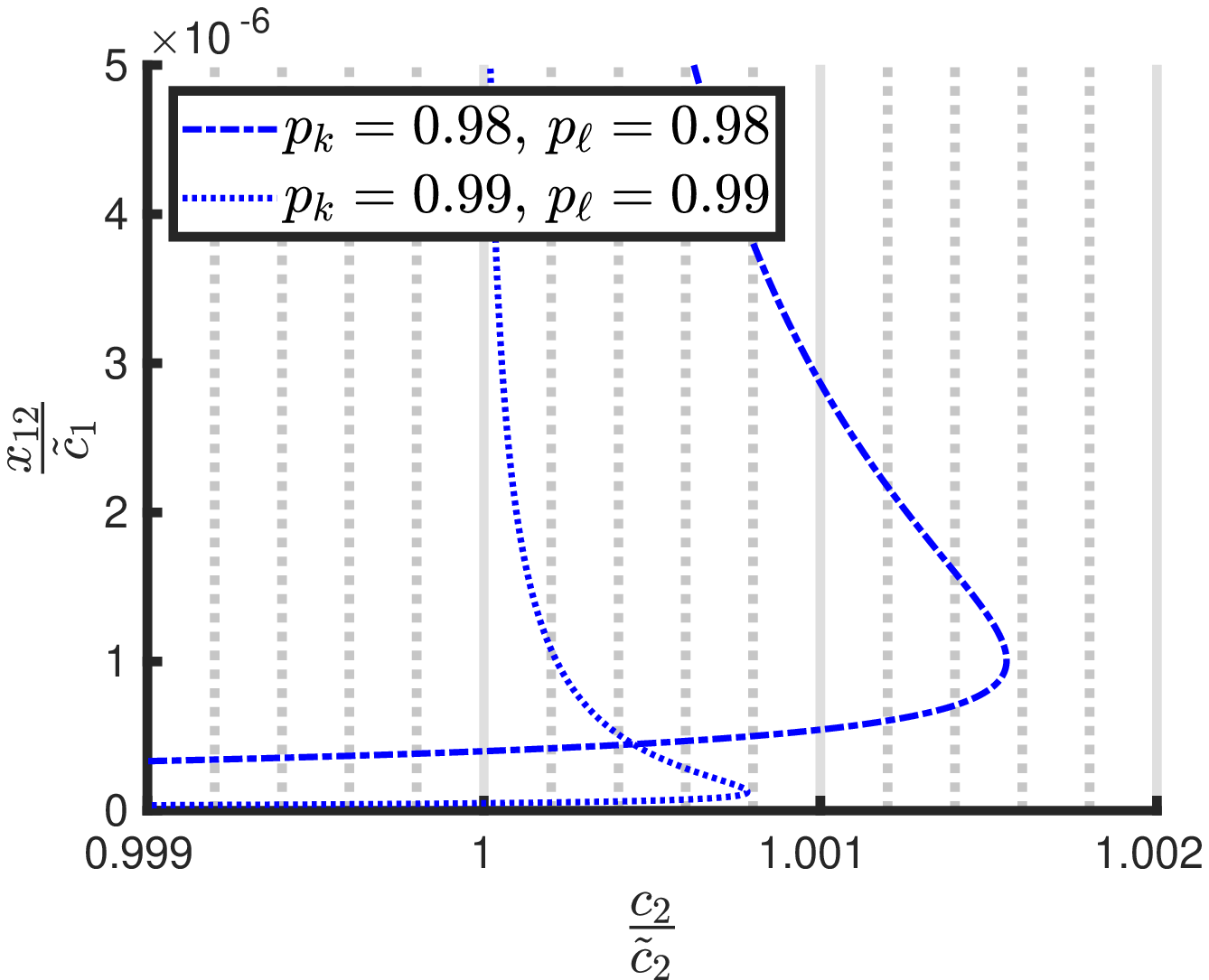}
    \subcaption{\label{fig:pk_pl_099_zoom22}$0.99\leq \frac{c_2}{\tilde c_2}\leq 1.01$}
  \end{subfigure}

  \caption{\label{fig:pk_pl_099}
    Numerical investigation of multistationarity for $p_k = p_\ell$ close to 1.
    An increase in $p_k$ and $p_\ell$ leads to a
    decrease in
    $\frac{x_{12}}{\tilde c_1}$ at $\frac{c_2}{\tilde c_2}
    \approx 1.25$ (display~\subref{fig:pk_pl_09_full2}).
    Also,
    when there is multistationarity,
    the values of $\frac{x_{12}}{\tilde c_1}$ (at all three steady states) decrease (possibly approaching 0) as $p_k$ and $ p_\ell$ approach 1.
    (display~\subref{fig:pk_pl_099_zoom12}).
    Finally, as $p_k$ and $ p_\ell$ approach 1, the multistationarity interval becomes so small that
    the curve approaches a step function
    (displays~\subref{fig:pk_pl_09_full2}--\subref{fig:pk_pl_099_zoom22}).
  }
\end{figure}




{\color{black}
\begin{remark}
\label{rem:ref-1}
    In the limiting case of $p_k\to1$ and $p_\ell\to1$, multistationarity deforms to monostationarity.
    It would be interesting to investigate what happens to the steady states; for instance, do two of them merge to form one?
    One setup for studying this in a controlled way is to fix $\kcat$ and $\lcat$, and then let $\koff$ and $\loff$ go to 0.
\end{remark}
}

\section{Hopf bifurcations and oscillations} \label{sec:osc}
In this section, we investigate Hopf bifurcations and oscillations in the reduced ERK network.
%
%
First, we answer Question~\ref{q:reduced-osc} in the affirmative:
Theorem~\ref{thm:pre-hopf-all-process-levels-epsilon-close-to-1}
asserts
that a Hopf bifurcation exists 
at all processivity levels $p_k$ and $p_{\ell}$ arbitrarily close to 1 -- and in fact for all levels greater than 0.003.
Subsequently, we perform a numerical investigation into oscillations.

\begin{theorem}[Hopf bifurcations at all processivity levels] \label{thm:pre-hopf-all-process-levels-epsilon-close-to-1}
Consider the reduced ERK network.
For all $0.002295< \epsilon < 1$, there exists a rate-constant vector
$\kappa^* = ( k_1^*, k_3^*, \kcat^*,\koff^*, m^* ,\ell_1^*, \ell_3^*, \lcat^*,\loff^*, n^*)$
such that
    \begin{enumerate}
        \item $p_k={\kcat^*}/{(\kcat^*+\koff^*)}>\epsilon$ and
         $p_{\ell} ={\lcat^*}/{(\lcat^* +\loff^*)} > \epsilon$, and
        \item         the resulting system \eqref{eq:ODE-reduced} admits a simple Hopf bifurcation (with respect to $\kcat$).
    \end{enumerate}
\end{theorem}

\begin{proof}
Fix $0.002295< \epsilon < 1$.
Observe that, for every choice of rate constants for which
(a) $\kcat^*>\epsilon/(1-\epsilon)> 0.002295/(1-0.002295)\approx 0.0023$,
(b) $\lcat^*:= t^2 \kcat^* $ (for any choice of $t>1$), and
(c) $\koff^*= \loff^* := 1$,
we obtain the desired inequalities for $p_k$ and $p_{\ell}$:
\begin{align} \label{eq:bounds}
 \epsilon ~<~ \dfrac{\kcat^*}{\kcat^*+1} ~=~ p_k ~ <~  \dfrac{t^2 \kcat^* }{t^2 \kcat^* + 1} ~=~ p_{\ell}~.
\end{align}

Next, we show that a Hopf bifurcation exists, by verifying the conditions on $\h_4$, $\h_5$, and $\h_6$ (as in Proposition~\ref{prop:hopfcriterion}).
First, we show in the supplementary file {\tt redERK-Hopf.mw} that
$\h_4(\hat{\kappa};~x)$ is a sum of positive terms, and thus $\h_4(\hat{\kappa};~x)>0$ for all $\hat{\kappa} = (\kcat, \koff, \loff) \in \mathbb{R}^3_{>0}$ and $x \in \mathbb{R}^{10}_{>0}$.

Next, let
$(\widehat{\kappa};~x):= 
(\kcat^* , 1, 1;~ 1,
    1,
    1,
    t^2,
    1,
    t^2,
    1/t,
    1,
    t^2,
    1).
$
We verify (using {\tt Mathematica}) that if $\kcat^*>0.0023$,
then $\h_5(\hat{\kappa}^*;x)>0$ for all $t>0$;
see the supplementary file {\tt h5pos.nb}. Fix $\kcat^*>0.0023$.
Substituting $t^*=1$ into $\h_6(\widehat{\kappa}^*;x^*)$ yields a positive polynomial (in $\kcat^*$):
{\scriptsize
\begin{align*}
\h_6(\widehat{\kappa}^*;x^*)|_{t^*=1} =
({\kcat^*} +1)^2 \bigg(&31824000 {\kcat^*}^{18}+713988320 {\kcat^*}^{17}+7660517072 {\kcat^*}^{16}+52115784592 {\kcat^*}^{15} +251452795392 {\kcat^*}^{14}\\
&+912214161728 {\kcat^*}^{13}+2574990720896 {\kcat^*}^{12}+5775757031984 {\kcat^*}^{11}+10424374721840 {\kcat^*}^{10}\\
&+15237491111424 {\kcat^*}^{9} +18065664178000 {\kcat^*}^{8}+17318286301088 {\kcat^*}^{7}+13314668410544 {\kcat^*}^{6} \\
&+8093460125184 {\kcat^*}^{5}+3802097816832 {\kcat^*}^{4}+1331324403072 {\kcat^*}^{3} +327072356352 {\kcat^*}^{2}\\
&+50292006912 {\kcat^*} +3641573376\bigg)~.
\end{align*}
}%
Also, as $t\rightarrow \infty$, the limit of $\h_6(\hat{\kappa}^*;x^*)$ is $-\infty$.
Hence, there exists $t^*>1$ such that
$\h_6(\widehat{\kappa}^*;x^*)=0$ (where $x^*= \left(1,
    1,
    1,
    {t^*}^2,
    1,
    {t^*}^2,
    1/{t^*},
    1,
    {t^*}^2,
    1\right)$);
see the supplementary file {\tt redERK-Hopf.mw}. Finally, we check that $\pd{\h_6}{\kcat}(\widehat{\kappa}^*;x^*)\neq 0$ whenever $\h_6(\widehat{\kappa}^*;x^*)=0$ --
we verified this using the {\tt Julia} package {\tt HomotopyContinuation.jl}  \citep{HomotopyContinuation.jl} (see the supplementary file {\tt nondegen-close-to-1.txt}).

Thus, the reduced ERK system admits a Hopf bifurcation at
\begin{align} \label{eq:x-star}
    x^*~:=~ (x^*_1, x^*_2,\dots, x^*_{10} ) ~=~
\left(    1,
    1,
    1,
    {t^*}^2,
    1,
    {t^*}^2,
    1/t^*,
    1,
    {t^*}^2,
    1\right)~,
\end{align}
when the rate-constant vector is
\begin{align} \label{eq:kappa-star}
    \kappa^* ~:&=~ ( k_1^*, k_3^*, \kcat^*,\koff^*, m^* ,\ell_1^*, \ell_3^*, \lcat^*,\loff^*, n^*)  \\
    ~&=
    \left((\kcat^*+1){t^*}^2,
    (\kcat^*+1){t^*}^2,
    \kcat^*,
    1,
    t^*,
    \kcat^* {t^*}^2+1,
    (\kcat^* {t^*}^2+1)/{t^*}^2,
    \kcat^* {t^*}^2,
    1,
    1\right)~. \notag
\end{align}
By construction, these rate constants satisfy the conditions
(a), (b) (with $t=t^*>1$), and~(c) listed at the beginning of the proof.
So, the inequalities~\eqref{eq:bounds} hold.
\end{proof}

\begin{remark}
Following the proof of Theorem~\ref{thm:pre-hopf-all-process-levels-epsilon-close-to-1},
we provide witnesses for the Hopf bifurcation for several values of
$p_k$ and $p_{\ell}$
in the supplementary file {\tt redERK-Hopf.mw} (under the ``First Vertex Analysis'' section) for the interested reader. 
For instance, when $\epsilon=0.89$, then then the choices $\kcat^*=9$ and $t^*\approx 124.02$
satisfy the conditions imposed in the proof, and so we obtain, as in~\eqref{eq:bounds},
the processivity levels $p_k=0.9$
and
$\pl \approx 0.999993$. 
Thus, from (\ref{eq:x-star}), there is a Hopf bifurcation at
$    x^*
\approx (
1,
1,
1,
15380.68,
1,
15380.68,
0.008,
1,
15380.68,
1
)
$ when the rate-constant vector is as in (\ref{eq:kappa-star}):
\begin{align*}
    \kappa^*
    ~&~\approx(153806.78,
    153806.78,
    9,
    1,
    124.02,
    138427.1,
    9.00,
    138426.11,
    1,
    1)~.
\end{align*}
\end{remark}

\begin{remark}[Relation to Question~\ref{q:main}]\label{rmk:reln-to-orig-q}
As noted earlier, Theorem~\ref{thm:pre-hopf-all-process-levels-epsilon-close-to-1}
addresses
Question~\ref{q:reduced-osc},
the reduced-ERK version of the original Question~\ref{q:main}.
We focused on the reduced ERK network rather than the original ERK network, because analyzing the original one is computationally challenging.

Nevertheless, we conjecture that
Theorem~\ref{thm:pre-hopf-all-process-levels-epsilon-close-to-1} ``lifts'' to the original ERK network.  Indeed,
to go from the reduced ERK network to the original ERK network, we make some reactions reversible (which is known to preserve oscillations~\citep{banaji-inheritance}) and add some intermediate complexes (which is conjectured to preserve oscillations~\citep{banaji-inheritance}).
More precisely, we hope for a future result that states that
adding intermediates preserves oscillations and Hopf bifurcations, while the ``old'' rate constants are only slightly perturbed.  Such a result would help us to elevate Theorem~\ref{thm:pre-hopf-all-process-levels-epsilon-close-to-1} to an answer to Question~\ref{q:main} for the original ERK network.  {\color{black} An approach to achieving such a result is to use the results of
\cite{feliu2019quasi}
to write the reduced system as a limiting case of the original system, where some parameter goes to zero,
and then give an argument like that in \cite[\S 3]{sustained}.
}
\end{remark}

\begin{remark} \label{rmk:bound}
The bounds  $p_k, p_{\ell} > 0.002295$ in Theorem~\ref{thm:pre-hopf-all-process-levels-epsilon-close-to-1} arose from our choice of specialization in the proof, namely, $(\widehat{\kappa};~x):= 
(\kcat^* , 1, 1;~ 1,
    1,
    1,
    t^2,
    1,
    t^2,
    1/t,
    1,
    t^2,
    1)$.  Another specialization (that admits a Hopf bifurcation) would give rise to other bounds on $p_k$ and $p_{\ell}$.  Nevertheless, as our interest is in $p_k$ and $p_{\ell}$ close to 1, our bounds are not restrictive.
\end{remark}

Next, we relax the hypothesis $p_k > 0.002295$ in
Theorem~\ref{thm:pre-hopf-all-process-levels-epsilon-close-to-1}
to allow for all values of $p_k>0$.  However, we cannot also simultaneously control $p_{\ell}$.

\begin{proposition}[Hopf bifurcations at all $p_k$] \label{prop:hopf-at-all-pk}
Consider the reduced ERK network.
For every choice of processivity level
$p_k \in (0,1)$,
there exists a rate-constant vector
$\kappa^* = ( k_1^*, k_3^*, \kcat^*,\koff^*, m^* ,\ell_1^*, \ell_3^*, \lcat^*,\loff^*, n^*)$
such that
    \begin{enumerate}
        \item $p_k={\kcat^*}/{(\kcat^*+\koff^*)}$, and
        \item the resulting system admits a Hopf bifurcation.
    \end{enumerate}

\noindent Moreover, by symmetry of $\kcat$ and $\lcat$ in the reduced ERK network, we have the analogous result for all choices of $p_{\ell}$.
\end{proposition}

\begin{proof}
As in the proof of Theorem~\ref{thm:pre-hopf-all-process-levels-epsilon-close-to-1},
we achieve any desired value of $p_k \in (0,1)$ by setting
$\koff^*=1$ and $\kcat^* = p_k/(1-p_k)$. Accordingly, consider any $\kcat^*\in \R_{>0}$.
We will show, using Proposition~\ref{prop:hopfcriterion}, that there exists $t^*>0$ such that the reduced ERK network admits a Hopf bifurcation at
\begin{align*}
    x^*~:=~ (x^*_1, x^*_2,\dots, x^*_{10} ) ~=~
\left(    1,
    1,
    1,
    1/{t^*}^2,
    1,
    1,
    t^*,
    1,
    1/{t^*}^2,
    1\right)~,
\end{align*}
when the rate-constant vector is
\begin{align*}
( &k_1^*, k_3^*, \kcat^*,\koff^*, m^* ,\ell_1^*, \ell_3^*, \lcat^*,\loff^*, n^*)\\
&= \left((\kcat^*+1)/{t^*}^2, (\kcat^*+1)/{t^*}^2, \kcat^*, 1, 1/{t^*}, ({t^*}^2+\kcat^*)/{t^*}^2, ({t^*}^2+\kcat^*)/{t^*}^4, \kcat^*/{t^*}^2, 1,  1/{t^*}^2\right).
\end{align*}
Indeed, we verify in the supplementary file {\tt redERK-Hopf-all-pk-values.mw} that
$\h_4(\widehat{\kappa};~x)>0$ and $\h_5(\widehat{\kappa};~x)>0$ for all $\hat{\kappa}= (\kcat, 1,1) \in \mathbb{R}^3_{>0}$ and $x = (1,1,1,x_4,1,1,x_7,1,x_9,1) \in \mathbb{R}^{10}_{>0}$, and that
$\h_6(\widehat{\kappa}^*;~x^*)=0$ for some $t^*>0$. Finally, in the supplementary file {\tt nondegen-all-process.txt}, we show that $\pd{\h_6}{\kcat}(\widehat{\kappa}^*;x^*)\neq 0$ whenever $\h_6(\widehat{\kappa}^*;x^*)=0$.
\end{proof}

We end this section with a numerical investigation into the effect of processivity levels on oscillations arising from the Hopf bifurcations analyzed above.
    Again we focus on the concentration of the fully phosphorylated substrate, in this case $x_5$.
We see in Figure~\ref{fig:osci} that indeed processivity levels have a large effect on the dynamics: as $p_k$ and $p_\ell$ approach 1, the amplitude decreases while the period increases -- at least for
the rate-constant vectors $\kappa^*$ and initial conditions we investigated (see the caption of Figure~\ref{fig:osci}).
It is an interesting question whether or not this phenomenon arises at other regions of parameter space.
We conjecture that indeed oscillations always dampen as as $p_k$ and $p_\ell$ approach 1.



\begin{figure}
    \centering
    \includegraphics[width=\linewidth]{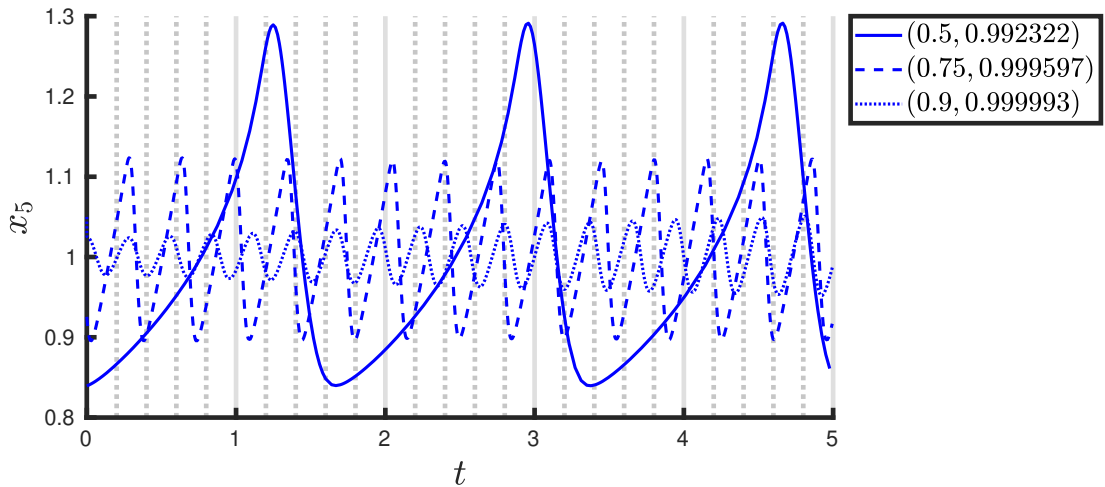}
    \caption{\label{fig:osci}
        For the reduced ERK network,
        oscillations in $x_5$ arising from three pairs of processivity levels $(p_k, p_{\ell})$.
        The rate-constant vectors $\kappa^*$ were obtained from~\eqref{eq:kappa-star}, using the values in  Table~\ref{tab:osci_table}.
        The initial conditions were chosen to be close to -- and in the same compatibility class as -- the corresponding Hopf bifurcation $x^*$ from \eqref{eq:x-star} (using the values in Table~\ref{tab:osci_table}); specifically, we perturbed $x^*$ by adding 0.05 to $x^*_5$ and subtracting $0.05$ from $x_6^*$.
    }
\end{figure}

\begin{table}[th]
    \centering
    \begin{tabular}{|c|c|c|c|}\hline
         $\kcat^*$ & $t^*$ & $p_k$ & $p_\ell$ \\ \hline \hline
         1 & 11.3685 & 0.5 & 0.992322 \\ \hline
         3 & 28.7451 & 0.75 & 0.999597 \\ \hline
         9 & 130.22 & 0.9 & 0.999993 \\ \hline
    \end{tabular}
    \caption{Values of $\kcat^*$ and $t^*$ used for Figure~\ref{fig:osci}, and resulting processivity levels, as in~\eqref{eq:bounds}.}
    \label{tab:osci_table}
\end{table}

\section{Coexistence of bistability and oscillations} \label{sec:coexist}
Having shown that multistationarity and Hopf bifurcations exist in certain ERK systems for (nearly) all possible processivity levels, we now investigate whether these two dynamical phenomena can occur together.
The first question is whether
bistability and oscillations
can coexist in the same compatibility class (Section~\ref{sec:preclude-coexistence}), and then we consider coexistence in distinct compatibility classes (Section~\ref{sec:coexist-distinct-classes}).

\subsection{Precluding coexistence in a compatibility class} \label{sec:preclude-coexistence}
The next result, which {\color{black} applies to general networks,}
forbids bistability and Hopf bifurcations from occurring in the same compatibility class, when there are up to 3 steady states and certain other conditions are satisfied.
These conditions allow us to apply (in the proof) results from degree theory.

\begin{theorem} \label{thm:no-coexistence}
Consider a reaction system $(G,\kappa)$.  Let $\mathcal{S}_c$ be a compatibility class such that (1) the system is dissipative\footnote{{\em Dissipative} means that there is a compact subset of $\mathcal{S}_c$ that every trajectory eventually enters;
being dissipative is automatic when the network is conservative~\citep{CFMW}.}
with respect to $\mathcal{S}_c$, and (2) $\mathcal{S}_c$ contains at most 3 steady states and no boundary steady states.
Then $\mathcal{S}_c$ does \underline{not} contain both a simple Hopf bifurcation and two stable steady states.
\end{theorem}

\begin{proof}
Let $W$ be a $d \times s$ (row-reduced) conservation-law matrix, where $d$ is the number of conservation laws and $s$ is the number of species.
Let $f_{c,\kappa}$ be the resulting augmented system.

We examine, for certain $x^*$ in $\mathcal{S}_c$,
the coefficient of $\lambda^d$ in $\det(\lambda I - {\rm Jac} f)|_{x=x^*}$.
If $x^*$ is a Hopf bifurcation, then (by
a criterion of \cite{yang-hopf},
restated in~\citep[Proposition 2.3]{mixed}) the coefficient is positive.
Similarly, if
$x^*$ is a stable steady state, then (by the Routh-Hurwitz criterion) the coefficient
is positive.
Finally, 
by a straightforward generalization of~\citep[Proposition 5.3]{wiuf-feliu-power-law},
the coefficient equals $(-1)^{s-d} \det {\rm Jac} f_{c,\kappa}|_{x=x^*}$.

Assume for contradiction that $\mathcal{S}_c$ contains a simple Hopf bifurcation $x^{(1)}$ and two stable steady states $x^{(2)}$ and $x^{(3)}$ (and hence no more steady states by hypothesis).  Then (by definition~\citep{CFMW} and by above) the Brouwer degree of $f_{c,\kappa}$ with respect to $\mathcal{S}_c$ is as follows:
    \begin{align*}
            {\rm sign} \det {\rm Jac} f_{c,\kappa}|_{x=x^{(1)}}
    +
    {\rm sign} \det {\rm Jac} f_{c,\kappa}|_{x=x^{(2)}}
    +
    &
    {\rm sign} \det {\rm Jac} f_{c,\kappa}|_{x=x^{(3)}} \\
    & \quad
    ~=~
    (-1)^{s-d}
    +
    (-1)^{s-d}
    +
    (-1)^{s-d}~,
    \end{align*}
which yields a contradiction, as the degree must be $\pm 1$ (see~\citep{CFMW}).
\end{proof}

For the minimally bistable ERK subnetwork, Theorem~\ref{thm:no-coexistence}
implies that,
{\em if the following conjecture holds, Hopf bifurcations and bistability do \underline{not} coexist in compatibility classes}:
\begin{conjecture} \label{conj:min-bistab-at-most-3}
For the minimally bistable ERK subnetwork, the maximum number of positive steady states (in any compatibility class, for any choice of rate constants) is 3.
\end{conjecture}

\noindent
The maximum number of positive steady states is at most 5~\citep{OSTT}, and a version of this conjecture was stated earlier (see \citep[Propositions~5.8--5.9 and Conjecture 5.10]{OSTT}).
We pursue the conjecture in Section~\ref{sec:number-steady-states}.

\subsection{Coexistence in distinct compatibility classes} \label{sec:coexist-distinct-classes}

Theorem~\ref{thm:no-coexistence} precludes, for certain reaction systems, the coexistence of
bistability and a simple Hopf bifurcation in a single compatibility class.
Next, for ERK systems, we ask about coexistence in {\em distinct} compatibility classes.

\begin{question} \label{q:coexistence-diff-classes}
Is it possible in one of the ERK networks (the original one or the minimally bistable ERK subnetwork\footnote{The reduced ERK network is not in this list, as it does not admit bistability~\citep{OSTT}.})
to have -- for some choice of positive rate constants --
2 stable steady states in one compatibility class
and a simple Hopf bifurcation in another? 
\end{question}

As an initial investigation
we examine the minimally bistable ERK network (see the  supplementary file {\tt min-bistab-ERK-Hopf-and-Bistability.mw}).
This network yields a Hopf bifurcation when $\kon= 4.0205$ and the other rate constants are as in~\citep[Equation~(23)]{OSTT} (these non-$\kon$ rate constants yield oscillations in the fully irreversible ERK network).  However, for this choice of rate constants, there is no bistability (in any compatibility class), which we determined by computing the critical function, much like in the proof of~\citep[Proposition 4.5]{OSTT}.

\section{Maximum number of steady states} \label{sec:number-steady-states}
In this section, we pursue Conjecture~\ref{conj:min-bistab-at-most-3},
which states that the maximum number of positive steady states of the minimally bistable ERK subnetwork is 3.
The idea is first to reduce to a system of 3 equations in 3 variables (Proposition~\ref{prop:3-eqns}) and then, using resultants, to further reduce to a single univariate polynomial
(Proposition~\ref{prop:resultant}).

Our methods are similar to
the approach that \cite{WangSontag}
took to analyze the fully distributive, dual-site phosphorylation system.
Namely, we substitute a steady-state parametrization from~\citep{OSTT} for the minimally bistable ERK subnetwork into the conservation laws, which yields a polynomial system in only 3 variables.  We then show that the maximum number of positive roots of this family of polynomial systems is equal to the maximum number of steady states (as in Conjecture~\ref{conj:min-bistab-at-most-3}).

\begin{proposition} \label{prop:3-eqns}
Consider the family of polynomial systems in $x_1,x_2,x_3$ given by:
{\small
\begin{eqnarray} \label{eq:3-eqns}
    \notag
    c_1-c_2-c_3
    & = &
     x_1-x_2-x_3+\frac{a_5 a_9 a_{10} x_1 x_2}{a_8 x_2+a_{13} x_3+a_4 a_9 a_{13} x_3}
    +\frac{a_5a_7a_{10}x_1(a_8 x_2+a_{13} x_3)}{ a_1 a_{11}( a_8 x_2+a_{13} x_3+a_4 a_9 a_{13} x_3)}
     \\
    & & 
    +\frac{a_5 a_{10} x_1 x_2 (a_8 x_2+a_2 a_7 a_8 x_2+ a_{13} x_3+a_2 a_7 a_{13} x_3)}{a_1a_3 a_{12}  x_3 (a_8 x_2+a_{13} x_3+a_4 a_9 a_{13} x_3)} ~,  \\
    \label{eq:3-eqns-2}
	c_2 & = &
     x_2+\frac{  a_5a_{10} x_1x_2(a_8 x_2+a_{13} x_3)}{
     a_8 x_2+a_{13} x_3+a_4 a_9 a_{13} x_3}
     +\frac{a_5a_7a_{10}  x_1x_2(a_8 x_2+a_{13} x_3) }{a_1(a_8 x_2+a_{13} x_3+a_4 a_9 a_{13} x_3)}+a_{10} x_1 x_2 ~,
    \\
	\notag
	c_3 & = &
    x_3+
    \frac{a_5a_{10}x_1x_2( a_8 x_2+a_2 a_7 a_8 x_2+a_{13} x_3+a_2 a_7 a_{13} x_3)}{a_1a_3(a_8 x_2+a_{13} x_3+a_4 a_9 a_{13} x_3)}
     \\  \label{eq:3-eqns-3}
    & & +\frac{a_5a_{10}x_1x_2(a_8 x_2+a_{13} x_3)}{a_1(a_8 x_2+a_{13} x_3+a_4 a_9 a_{13} x_3)}+\frac{a_5a_9a_{10}a_{13}x_1x_2x_3}{a_8 x_2+a_{13} x_3+a_4 a_9 a_{13} x_3}~,
\end{eqnarray}
}
where the coefficients $a_i$ and $c_i$ are arbitrary positive real numbers.  Then the maximum number of positive roots $x^* \in \mathbb{R}^3_{>0}$, among all such systems, equals the maximum number of positive steady states of the minimally bistable ERK network.
\end{proposition}

\begin{proof}
The equations~\eqref{eq:3-eqns}--\eqref{eq:3-eqns-3}  are obtained as follows.
Using the ``effective steady-state function'' $h_{c,a}$ from~\citep[Proposition 3.1]{OSTT}, we solve for $x_4,x_5,\dots, x_{12}$ in terms of $x_1,x_2,x_3$ (and the $a_i$'s), and then substitute the resulting expressions into the conservation equations~\eqref{eq:cons-law-irre},
except we replace the first conservation equation by the first one minus the sum of the second and third.
Now the result follows from the definition of ``effective steady-state function'' \citep{DPST, OSTT}.
\end{proof}

Next, we go from the 3 equations (in $x_1,x_2, x_3$) in~\eqref{eq:3-eqns}--\eqref{eq:3-eqns-3} to 2 equations (in $x_2$ and $x_3$), as follows.
All 3 equations in~\eqref{eq:3-eqns}--\eqref{eq:3-eqns-3} are linear in $x_1$, so we solve each for $x_1$, obtaining equations of the form
$x_1=\gamma_1(x_2,x_3)$,
$x_1=\gamma_2(x_2,x_3)$, and
$x_1=\gamma_3(x_2,x_3)$,
respectively.
Now, let
$g_1:=\gamma_3 - \gamma_2$ and
$g_2:=\gamma_1 - \gamma_2$.
These $g_i$'s are polynomials in $x_2$ and $x_3$ (with coefficients which are polynomials in the $a_i$'s and $c_i$'s).
By construction, and by Proposition~\ref{prop:3-eqns}, we immediately obtain the following result:

\begin{proposition} \label{prop:2-eqns}
Let $g_1$, $g_2$, and $\gamma_1$ be as above.
Then for the system $g_1=g_2=0$
(where the coefficients $a_i$ and $c_i$ are arbitrary positive real numbers),
the maximum number of positive roots $(x_2^*, x_3^*) \in \mathbb{R}^2_{>0}$ with $\gamma_1(x_2^*, x_3^*)>0$,
is equal to the maximum number of (positive) steady states of the minimally bistable ERK network.
\end{proposition}

Let $R$ be the resultant~\citep{cox-little-oshea} of $g_1$ and $g_2$, with respect to $x_2$ (this resultant is shown in the supplementary files {\tt maxNUMss.mw} and {\tt resultant.txt}).  We apply a standard argument using resultants to obtain the following result:

 \begin{proposition}\label{prop:resultant}
    Let
    $(a^*;c^*) = (a_1^*, \dots,a_{13}^*,c^*_1,c^*_2, c^*_3) \in \mathbb{R}^{16}_{>0}$.
    Let $R$ be as above.
    If the univariate polynomial $R|_{(a^*;c^*)}$
    has at most
    3 
    roots 
    in the interval $(0,\min \{c_1,c_3\})$, and if
    for every $x_3^* \in \mathbb{R}_{>0}$, the equation $g_1(x_2, x_3^*)|_{(a^*;c^*)}=0$ has at most one positive solution for $x_2$,
    then system~\eqref{eq:3-eqns}--\eqref{eq:3-eqns-3}, when specialized at $(a^*;c^*)$,  has at most 3 positive roots $x^* \in \mathbb{R}^3_{>0}$.
 \end{proposition}

    \begin{proof}
    By \citep[Page 163, Chapter 3,  Sec. 6, Proposition 1(i)]{cox-little-oshea},
     \begin{align} \label{eq:clo-1}
             R ~\in~ \langle g_1, g_2\rangle\cap {\mathbb Q}[a_1,a_2,\dots, a_{13}, c_1,c_2,c_3 , x_3]~.
      \end{align}
     By \citep[Page 125, Chapter 3, Sec. 2, Theorem 3(i)]{cox-little-oshea},
    \begin{align} \label{eq:clo-2}
        \overline{\pi\left({\mathcal V} \left(g_1, g_2\right)\right)} ~=~ {\mathcal V} \left( ~ \langle g_1, g_2\rangle\cap {\mathbb Q}[a_1,a_2,\dots, a_{13}, c_1,c_2,c_3 , x_3] ~ \right)~,
            \end{align}
    where $\pi: \mathbb{C}^{18} \to \mathbb{C}^{17}$ denotes the standard projection 	
given by $(a; c; x_3, x_2) \mapsto   (a;c;x_3)$, ${\mathcal V}(\cdot)$ denotes zero set over $\mathbb{C}$ of a set of polynomials,
and $\overline{S}$ denotes the Zariski closure in ${\mathbb C}^n$ \citep[Chapter 4]{cox-little-oshea} of a subset $S \subseteq \mathbb{C}^n$.
    So, by~\eqref{eq:clo-1} and~\eqref{eq:clo-2},
    \[\overline{\pi\left({\mathcal V}\left(g_1, g_2\right)\right)} ~\subseteq~ {\mathcal V}\left(R\right)~.\]
    Thus, for a given $(a^*;c^*) \in \mathbb{R}^{16}_{>0}$, because $R|_{(a^*;c^*)}$ has at most $3$ positive roots $x_3$ in the interval $(0,\min \{c_1,c_3\})$,
    it follows that the solutions of
    the system $g_1|_{(a^*;c^*)}=g_2|_{(a^*;c^*)}=0$
    have up to 3 possibilities for $x_3$-coordinates
    in the interval $(0,\min \{c_1,c_3\})$.
    Next, we use the hypothesis that (for every $x_3^* \in \mathbb{R}_{>0}$) the equation $g_1(x_2, x_3^*)|_{(a^*;c^*)}=0$ has at most 1 positive solution for $x_2$,
    to conclude that $g_1|_{(a^*;c^*)}=g_2|_{(a^*;c^*)}=0$  has
    at most 3 positive solutions $(x_2, x_3) \in \mathbb{R}^2_{>0}$ with $ x_3 < \min \{c_1,c_3\}$.
    Thus, by construction of $g_1$ and $g_2$ (see the paragraph before Proposition~\ref{prop:2-eqns}),
    the original system~\eqref{eq:3-eqns}--\eqref{eq:3-eqns-3}, when specialized at $(a^*;c^*)$,  has at most 3 positive roots $x^* \in \mathbb{R}^3_{>0}$.
    \end{proof}

As an example of how we can use Proposition~\ref{prop:resultant} to tackle Conjecture~\ref{conj:min-bistab-at-most-3}, we next give two corollaries.
We hope to pursue this direction more in future work.
\begin{corollary} \label{cor:resultant-1}
For every choice of $c_1^*, c_2^*, c_3^*, a_9^*\in {\mathbb R}_{>0}$, if all other $a_i^*$'s are equal to $1$, then the (specialized at $(a^*;c^*)$) original system~\eqref{eq:3-eqns} has at most 3 positive roots $x^* \in \mathbb{R}^3_{>0}$.
\end{corollary}
\begin{proof}
To apply Proposition~\ref{prop:resultant}, we first show that
the univariate polynomial $R|_{(a^*;c^*)}$
    has at most
    3 positive
    roots $x_3$. When all $a_i^*$'s except $a_9^*$ are equal to $1$, then this specialized resultant (see the supplementary file {\tt maxNUMss.mw}) is as follows:
  \begin{align} \label{eq:resultant-1}
  R|_{(a^*;c^*)} ~=~ a_9^*x_3^2(a_9^*x_3+3c_2^*+3x_3)(C_4x_3^4+ C_3x_3^3 + C_2x_3^2 + C_1x_3+ C_0)~,
  \end{align}
  where
  \[C_4~=~2{a_9^*}^2+12a_9^*~, \quad \quad C_0~=~ -2 c_3^*(c_2^*-c_3^*)^2,\]
  and $C_1, C_2, C_3 \in {\mathbb Q}[a_9^*;c^*]$. By inspection, $C_4>0$ and $C_0 \leq 0$, for all $c_1^*, c_2^*, c_3^*, a_9^*\in {\mathbb R}_{>0}$.
  We consider two cases.
  If $C_0=0$, then $x_3=0$ is solution of $R|_{(a^*;c^*)}=0$, and so (because the ``relevant'' factor of $R|_{(a^*;c^*)}=0$ in~\eqref{eq:resultant-1} has degree four) $R|_{(a^*;c^*)}=0$ has at most 3 positive roots $x_3$. If $C_0<0$, then the sequence $C_4, C_3, C_2, C_1, C_0$ has at most $3$ sign changes, and so, by Descartes' rule of signs,  $R|_{(a^*;c^*)}=0$ has at most 3 positive roots $x_3$.

Second, we show that for every $x_3^* \in \mathbb{R}_{>0}$, the equation $g_1(x_2, x_3^* )|_{(a^*;c^*)}=0$ has at most one positive solution for $x_2$.
When all $a_i^*$'s except $a_9^*$ are equal to $1$,
we have
(see the supplementary file {\tt maxNUMss.mw}):
  \[g_1(x_2, x_3^*)|_{(a^*;c^*)} = 3x_2^2+(a_9^*x_3^*-3c_2^*+3c_3^*)x_2
  -x_3^*(x_3^*+c_2^*-c_3^*)(a_9^*+3)~.\]
Viewing $g_1(x_2, x_3^*)|_{(a^*;c^*)}$ as a polynomial in $x_2$, the leading coefficient is $3$, which is positive.  So, by Descartes' rule of signs, it suffices to show that either the constant term is non-positive or the coefficient of $x_2$ is positive.  In other words, we must show that if the constant term is positive, then the coefficient of $x_2$ is positive.  Indeed, if $-x_3(x_3+c_2^*-c_3^*)(a_9^*+3)>0$, then $c_3^*>c_2^*$, and so the coefficient of $x_2$ is
$a_9^*x_3^*-3c_2^*+3c_3^*=a_9^*x_3^* + 3 (c_3^* - c_2^*)>0$.
%

By the above two steps and Proposition \ref{prop:resultant}, we conclude that
the system~\eqref{eq:3-eqns}
-- when specialized at $(a^*;c^*)$ --
has at most 3 positive roots $x^* \in \mathbb{R}^3_{>0}$.
\end{proof}

\begin{corollary} \label{cor:resultant-2}
 For every choice of $c_1^*, c_3^* \in \mathbb{R}_{>0}$, if
 \begin{itemize}
 \item[(i)]$a^*_9$ and $c_2^*$ are sufficiently large,
 \item[(ii)]all other $a_i^*$'s are equal to the same value $b$ and are sufficiently large, and also
 \item[(iii)]$b>c_2^*/c_3^*>1$ and $c_2^*>c_3^*+1$,
 \end{itemize}
    then the (specialized at $(a^*;c^*)$) original system~\eqref{eq:3-eqns}--\eqref{eq:3-eqns-3} has at most 3 positive roots $x^* \in \mathbb{R}^3_{>0}$.
\end{corollary}
\begin{proof}

First, we show that
the univariate polynomial $R|_{(a^*;c^*)}$
    has at most
    3 positive
    roots $x_3$.
    When all $a_i^*$'s except $a_9^*$ are equal to $b$, then (see {\tt maxNUMss.mw}) we have:
  \begin{align} \label{eq:resultant-2}
  R|_{(a^*;c^*)}~=~
  -\Sigma \cdot  (C_5x_3^5+C_4x_3^4+ C_3x_3^3 + C_2x_3^2 + C_1x_3+ C_0)~,
  \end{align}
  where $ \Sigma = b^{17}a_9^*x_3^2(2bc_2^*+c_2^*+a_9^*bx_3+2bx_3+x_3)$ (which is positive), and
  \[C_5~=~2a_9^*b^{5}(b-1)(b+1)(a_9^*b+2b+1)~,\]
  \begin{align*}
        C_1~&=~c_3^*(-a_9^*{c_2^*}^2-{c_2^*}^2-3a_9^*c_2^*c_3^*+2a_9^*c_1^*c_2^*-2a_9^*c_1^*c_3^*+4a_9^*{c_3^*}^2+c_1^*c_2^*-2c_1^*c_3^*+c_2^*c_3^*+2{c_3^*}^2)b^{7} \\
        & \quad +~ \text{lower-order terms in } b~, \\
         &=~     c_3^*(-a_9^*{c_2^*}^2 + ~ [\text{lower-order terms in $a_9^*$ and $c_2^*$}] )b^{7} ~+~ \text{lower-order terms in } b~,
        \end{align*}
  \begin{align*}
  C_0 ~&=~-c_3^*(b^2+1)(c_2^*-c_3^*)(a_9^*b^4c_3^*-a_9^*b^3c_2^*+a_9^*b^3c_3^*-a_9^*b^2c_3^*-b^3c_2^*+2b^3c_3^*-b^2c_3^*-bc_3^*+c_2^*) \\
  &=~ -b^9c_3^*(b^2+1) (c_2^*-c_3^*) \left(a_9^* b^3 (bc_3^* -c_2^*) +~ [\text{lower-order terms in $a_9^*$, $b$, $c_2^*$}] \right)~,
  \end{align*}
  and $C_2, C_3, C_4 \in {\mathbb Q}[a_9^*;c^*]$.
  Assume that $a_9^*$, $b$, and $c_2^*$ are sufficiently large positive numbers.
  Assume also that $b>c_2^*/c_3^*>1$.
  Then, by inspection, $C_5>0$, $C_1<0$, and $C_0<0$. So the sequence $C_5, C_4, C_3, C_2, C_1, C_0$ has at most $3$ sign changes. Hence, Descartes' rule of signs implies that $R|_{(a^*;c^*)}=0$ has at most 3 positive roots $x_3$.

Second, we show that for every $x_3^*\in \mathbb{R}_{>0}$, $g_1(x_2, x_3^*)|_{(a^*;c^*)}=0$ has at most 1 positive solution for $x_2$.
When all $a_i^*$'s except $a_9^*$ are equal to $b$, then (see {\tt maxNUMss.mw})
\begin{align*}
    g_1(x_2, x_3^*)|_{(a^*;c^*)} ~&=~
    (b^4+b^3+b^2)x_2^2
        \\
    & \quad
    + (a_9^*b^4x_3^*-b^4c_2^*+2b^4c_3^*-b^4x_3^*-b^3c_2^*+b^3c_3^*-b^2c_2^*+b^2x_3^*)x_2
    \\
    & \quad
    - b^2x_3^*(a_9^*b^2c_2^*-a_9^*b^2c_3^*+a_9^*b^2x_3^*+b^2c_2^*-2b^2c_3^*+2b^2x_3^*+bc_2^*-bc_3^*+bx_3^*+c_2^*)
\end{align*}
In particular, the constant term can be rewritten and bounded above as follows, where we use the assumption that $c_2^*>c_3^*+1$:
\begin{align*}
    &-b^2x_3^* \left(
    [a_9^*b^2]
    [c_2^*- c_3^*+ x_3^* + c_2^*/a_9^*] -2b^2c_3^*+2b^2x_3^* + bc_2^*-bc_3^*+bx_3^*+c_2^* \right)
   \\
     \quad < \quad
    &-b^2x_3^* \left(
    [a_9^*b^2]
    +~ [\text{lower-order terms in $a_9^*$, $b$, $c_2^*$}]
    \right)~.
\end{align*}
Thus,  if $a_9^*$, $b$, and $c_2^*$ are sufficiently large (and $c_2^*>c_3^*+1$), then the constant term of
 $g_1(x_2, x_3^*)|_{(a^*;c^*)}$
is negative.
Also, the leading coefficient, $b^4+b^3+b^2$, is positive.
So, there is exactly 1 sign change in the sequence of coefficients, and hence, by Descartes' rule of signs, $g_1(x_2, x_3^*)|_{(a^*;c^*)}$ has at most 1 positive solution.

The above two steps and Proposition \ref{prop:resultant} together imply that
the (specialized at $(a^*;c^*)$) system~\eqref{eq:3-eqns} has at most 3 positive roots $x^* \in \mathbb{R}^3_{>0}$.
\end{proof}

\begin{remark} \label{rmk:relevant-factor-in-resultant}
In the two above proofs, we saw the (specialized) resultants~\eqref{eq:resultant-1} and~\eqref{eq:resultant-2} have some ``irrelevant'' factors (those that are always positive) and one ``relevant'' factor, such that the sign of the resultant equals the sign of the relevant factor.  This is true for the resultant, even before specialization; see the supplementary file {\tt maxNUMss.mw}.
\end{remark}

\section{Discussion} \label{sec:disc}
The motivating question for this work is Question~\ref{q:main}, which
pertains to the important problem of how bistability and oscillations emerge in ERK networks.
We essentially answered this question.  What ``essentially'' means here is that
we answered the question for some closely related ERK networks, and
only two conjectures (Conjecture~\ref{conj:bistab} and see also Remark~\ref{rmk:reln-to-orig-q}) -- which we believe to be true -- stand in the way of complete answers.

We also pursued two related topics, the coexistence of oscillations and bistability, and the maximum number of positive steady states.  We showed that if another conjecture we believe to be true
(Conjecture~\ref{conj:min-bistab-at-most-3})
holds, then Hopf bifurcations and bistability do not coexist in compatibility classes in the minimally bistable ERK subnetwork.  We then pursued Conjecture~\ref{conj:min-bistab-at-most-3} using resultants, achieving partial results and laying the groundwork for future progress on this conjecture.  This question of the maximum number of positive steady states is important -- it is one way to measure a network's capacity for processing information -- and we would like in the future some easy criterion for computing this number for phosphorylation and other signaling networks.

Finally, our interest in phosphorylation networks is due to their role in mitogen-activated protein kinase (MAPK) cascades, which enable cells to make decisions (to differentiate, proliferate, die, and so on)~\citep{mapk-bba}.  
We therefore want to understand which types of dynamics MAPK cascades and phosphorylation networks are capable of, as
bistability and oscillations
may
be used by cells to make decisions and process information~\citep{tyson-albert}.  For MAPK cascades, to quote from \cite{sun2014enhancement},
``By adjusting the degree of processivity in our model, we find that the MAPK cascade is able to switch among the ultrasensitivity, bistability, and oscillatory dynamical states''.
Our results here are complementary -- even while keeping the processivity levels constant (at any amount), the ERK network can switch between a range of dynamical behaviors, from bistability to oscillations via a Hopf bifurcation.

\subsection*{Acknowledgements} {\small
Part of this research was initiated at the Madison Workshop on Mathematics of Reaction Networks at the University of Wisconsin in 2018.
NO, AS, and XT were partially supported by the NSF (DMS-1752672).
CC was partially supported by the Deutsche Forschungsgemeinschaft, 284057449.
The authors thank Elisenda Feliu,
Henry Mattingly,
Stanislav Shvartsman, Sascha Timme, Ang\'elica Torres, and Emanuele Ventura
for helpful discussions.
{\color{black} The authors acknowledge three referees whose insightful comments helped strengthen this work.  In particular,
Remark~\ref{rem:ref-1} is inspired by the very detailed comments of one reviewer who suggested the limiting process described there. Ideas in Remark~\ref{rmk:reln-to-orig-q} are also due to this reviewer.
}
}

\bibliographystyle{plainnat}
\bibliography{erk.bib}

\appendix
\section{Files in the Supporting Information} \label{app:supp-files}
Table~\ref{tab:supp-files} lists the
files in the Supporting Information, and the result or section each file supports.
All files can be found at the online repository: \url{https://github.com/neeedz/COST}

\begin{table}[ht]
\centering
\begin{tabular}{lll}
\hline
Name & File type & Result or Section   \\
\hline
{\tt minERK-MSS-bistab.mw}         & {\tt Maple}
&  Theorem~\ref{thm:MSS} \\
{\tt minERK-MSS-bistab.mw} & {\tt Maple}
& Section~\ref{sec:bistab-evidence} \\
{\tt redERK-Hopf.mw}         & {\tt Maple}
&  Theorem~\ref{thm:pre-hopf-all-process-levels-epsilon-close-to-1} \\
{\tt h5pos.nb} & {\tt Mathematica}
& Theorem~\ref{thm:pre-hopf-all-process-levels-epsilon-close-to-1} \\
{\tt nondegen-close-to-1.txt} & {\tt Text*} & Theorem~\ref{thm:pre-hopf-all-process-levels-epsilon-close-to-1}\\
{\tt redERK-Hopf-all-pk-values.mw}          & {\tt Maple}
& Proposition~\ref{prop:hopf-at-all-pk}\\
{\tt nondegen-all-process.txt} & {\tt Text*} & Proposition~\ref{prop:hopf-at-all-pk}\\
{\tt min-bistab-ERK-Hopf-and-Bistability.mw} & {\tt Maple} & Section~\ref{sec:coexist-distinct-classes}\\
{\tt maxNUMss.mw} & {\tt Maple}
& Section~\ref{sec:number-steady-states} \\
{\tt resultant.txt} & {\tt Text}
& Section~\ref{sec:number-steady-states} \\
\hline
\end{tabular}
\caption{Supporting Information files and the results they support.
Here, {\tt Text*} indicates an output file from using the {\tt Julia} package {\tt HomotopyContinuation.jl}~\citep{HomotopyContinuation.jl}.
\label{tab:supp-files}}
\end{table}

\section{Procedure to study  multistationarity numerically}\label{app:num-mss}
Here we describe the procedure we used in Section~\ref{sec:numerics} for numerically studying multistationarity in the minimally bistable ERK network at various processivity levels $\pk$ and $\pl$.

We begin by mirroring the analysis of Section~\ref{sec:bistab-evidence}.  Specifically, we use the parameters given in~\eqref{eq:special-rate-constants}
to study the critical function $C(\kappa,\hat x)$ for $x_1=x_2=T$ and
$x_3=1$. Due to this choice of $\kappa$ and $\hat c$, the critical
function is a (rational) function of $p_k$, $p_\ell$, and $T$ only,
i.e., $C(\kappa,\hat x) \equiv C(p_k,p_\ell,T)$. The numerator is the following polynomial:
\begin{equation}
  \tiny
  \begin{split}
    \label{eq:num_C}
    q(p_k,p_\ell,T) &=
    -p_k \left((3-2 p_\ell) p_\ell+p_k \left(1-3 p_\ell+2
        p_\ell^2\right)\right) \\
    &\quad
    +\left(-5 p_\ell^2+p_k p_\ell (-9+11 p_\ell)+p_k^3 \left(-1+3
        p_\ell-2 p_\ell^2\right)-p_k^2 \left(3-3 p_\ell+p_\ell^2\right)
    \right) T \\
    &\quad
    +\left(-8 p_\ell^2+p_k p_\ell (-13+9 p_\ell)+p_k^3
      \left(-1+p_\ell-p_\ell^2\right)+p_k^2 \left(-6+2 p_\ell+7
        p_\ell^2\right)
    \right) T^2 \\
    &\quad
    +\left(-3 p_\ell^2-p_k p_\ell (5+8 p_\ell)+p_k^3 \left(-4+3
        p_\ell+p_\ell^2\right)+p_k^2 \left(-3-10 p_\ell+13
        p_\ell^2\right)
    \right) T^3 \\
    &\quad
    +p_k \left((5-8 p_\ell) p_\ell+3 p_k \left(1-5
        p_\ell+p_\ell^2\right)+p_k^2 \left(-7+4 p_\ell+2 p_\ell^2\right)
    \right) T^4 \\
    &\quad
    +p_k \left(3 p_\ell+p_k^2 \left(-3-3 p_\ell+2 p_\ell^2\right)-p_k
      \left(-2+p_\ell+4 p_\ell^2\right)
    \right) T^5
    -2 (-1+p_k) p_k^2 p_\ell T^6
  \end{split}
\end{equation}
As $0 < p_k, p_\ell < 1$, the leading coefficient of
$q(p_k,p_\ell,T)$ as a polynomial in $T$ is positive.
Next, the
steady-state parametrization $\phi$ from
Proposition~\ref{prop:param-irrev} is as follows (cf.\ eq.~(\ref{eq:param-irrev-details})):
\begin{gather}
  \notag
  x_1 = T,\; x_2 = T,\; x_3 =1,\;  x_{4} = \frac{p_k T^2 (1+T)}{p_\ell+p_k p_\ell T},\;
  x_{5} = -\frac{p_k (-1+p_\ell) T (1+T)}{p_\ell+p_k p_\ell T}, \\
  \label{eq:ss_para}
  x_{6} = \frac{T-p_k T}{1+p_k T},\;
  x_{7} = -\frac{(-1+p_k) T (1+T)}{1+p_kT},\;
  x_{8} = -\frac{p_k (-1+p_\ell) T (1+T)}{p_\ell+p_k p_\ell T},\\
  \notag
  x_{9} = \frac{T-p_k T}{1+p_k T},\;
  x_{10} = -\frac{p_k (-1+p_\ell) (1+T)}{p_\ell+p_k p_\ell T},\;
  x_{11} = T^2,\;
  x_{12} = \frac{p_k (1+T)}{p_\ell+p_k p_\ell T}
\end{gather}

To numerically study multistationarity for $p_k$, $p_\ell\to 1$, we proceed as follows:
\begin{enumerate}[{(}i{)}]
\item\label{item:pick_pk_pl}
    Pick values of $0 < \tilde p_k, \tilde p_\ell < 1$  and
  $\tilde T>0$ such that $q(\tilde p_k, \tilde p_\ell, \tilde T)>0$ (recall eq.~\eqref{eq:num_C}).
\item\label{item:compute_tilde_x}
    Substitute into~(\ref{eq:ss_para})
    the values of $\tilde p_k$, $\tilde p_\ell$,
  and $\tilde T$ from the previous step to
  obtain a steady state $\tilde x$.
\item\label{item:compute_tots}
    Compute, using (\ref{eq:cons-law-irre}),
    the total amounts $\tilde c_1$, $\tilde c_2$, and
  $\tilde c_3$ at $\tilde x$.
\item Use {\tt Matcont} with initial condition near  $\tilde x$ and bifurcation parameter $c_2$,
   to obtain a bifurcation curve. 
\item\label{item:normalize} To compare curves corresponding to distinct
  $\tilde p_k$ and $\tilde p_\ell$, compute relative concentrations
  $\frac{x_i}{\tilde c_1}$ and $\frac{c_2}{\tilde c_2}$ that
  relate $x_i$ and $c_2$ to the $\tilde x$ and $\tilde c_2$ computed
  in steps~\ref{item:compute_tilde_x}~and~\ref{item:compute_tots}.
\end{enumerate}
Step~\ref{item:normalize} is crucial for
interpreting the
numerical results obtained by
the above procedure,
because certain
total amounts differ by orders of
magnitude as $p_k$, $p_\ell \to 1$, and so it is more meaningful to compare
values relative to the reference point $\tilde x$ obtained in
step~\ref{item:compute_tilde_x}.

\begin{table}[htb]
  \centering
  \begin{tabular}{|c||c|c|c|c|c|c|c|c|c|c|c|c|}\hline
    & \multicolumn{4}{c|}{Figure~\ref{fig:pl_only}} &
    \multicolumn{5}{c|}{Figure~\ref{fig:pk_pl_09}}
    \\ \hline
    $p_k$ & 0.1 & 0.1 & 0.1 & 0.1 & 0.5 & 0.75 & 0.8 & 0.85 & 0.9 \\ \hline
    $p_\ell$ & 0.1  & 0.5 & 0.9 & 0.999 & 0.9 & 0.9 & 0.9 & 0.9  & 0.9
    \\ \hline
    $T$ &  2.21958 & 3.72221 & 4.98625 & 5.28023 & 4.80723 & 8.59917 &
    10.6576 & 14.1522 & 21.2341 \\\hline
  \end{tabular}
  \caption{Values of 
  $p_k$, $p_\ell$, and $T$ used in
    Figures~\ref{fig:pl_only} and \ref{fig:pk_pl_09}.}
  \label{tab:pk_pl_T_1_2}
\end{table}

\begin{table}[htb]
  \centering
  \begin{tabular}{|c||c|c|c|c|}\hline
    & \multicolumn{4}{c|}{Figure~\ref{fig:pk_pl_099}}
    \\ \hline
    $p_k$ & 0.9 & 0.95 & 0.98 & 0.99 \\ \hline
    $p_\ell$ &  0.9 & 0.95 & 0.98 & 0.99 \\ \hline
    $T$ &  21.2341 & 43.2027 & 109.186 & 219.18 \\ \hline
  \end{tabular}
  \caption{Values of 
  $p_k$, $p_\ell$, and $T$ used in
    Figure~\ref{fig:pk_pl_099}.
    \label{tab:pk_pl_T_3}
  }
\end{table}

\end{document}